\documentclass[onecolumn, 11 pt]{article}

\usepackage{lipsum}
\usepackage{epstopdf}
\usepackage{algorithmic}
\usepackage{amsmath, amsfonts, amssymb}
\usepackage{color}
\usepackage{graphicx}
\usepackage{stackrel}
\usepackage[normalem]{ulem}
\usepackage{tikz}

\usepackage{geometry}
\newenvironment{proof}{%
    \textit{Proof:}%
}{%
    \\
}

\newcommand{\E}{\mathsf{E}}
\newcommand{\R}{\mathbb{R}}
\newcommand{\N}{\mathcal{N}}
\newcommand{\eg}{{i.e.}, }
\newcommand{\I}{\mathfrak{I}}
\newcommand{\U}{\mathcal{U}}
\newcommand{\Q}{\mathcal{Q}}
\newcommand{\bbN}{\mathbb{N}}
\newcommand{\p}{\mathcal{P}}
\newcommand{\Y}{\mathcal{Y}}
\newcommand{\X}{\mathcal{X}}

\newcommand{\mO}{\mathcal{O}}
\newcommand{\W}{\mathcal{W}}
\newcommand{\V}{\mathcal{V}}

\newcommand{\T}{^{\mbox{\tiny \sf T}}}
\newcommand{\tr}{{\rm{tr}}}
\newcommand{\h}{\mathcal{H}}
\newcommand{\Z}{\mathcal{O}}
\newcommand{\xn}{X^{\rm{new}}}
\newcommand{\yn}{Y^{\rm{new}}}
\newcommand{\en}{e^{\rm{new}}}

\newtheorem{thm}{Theorem}[section]

\newtheorem{pr}[thm]{Proposition}
\newtheorem{rem}[thm]{Remark}

\newtheorem{lm}{Lemma}

\ifpdf
  \DeclareGraphicsExtensions{.eps,.pdf,.png,.jpg}
\else
  \DeclareGraphicsExtensions{.eps}
\fi

\title{ 
Optimal Quantizer Scheduling and Controller Synthesis for Partially Observable  Linear Systems\footnote{
{This work has been supported by ARL under DCIST CRA W911NF-17-2-0181 and by ONR award N00014-18-1-2375.}}}

\author{Dipankar Maity  \thanks{Department of Electrical and Computer Engineering, University of North Carolina at Charlotte, Charlotte, NC, 28223, USA,
  ({dmaity@uncc.edu}).}
\and Panagiotis Tsiotras\thanks{Guggenheim School of Aerospace Engineering, Georgia Institute of Technology, Atlanta, GA, 30332, USA,
  ({tsiotras@gatech.edu}).}
}

\date{}

\usepackage{amsopn}

\begin{document}

\maketitle

\begin{abstract}
 In networked control systems, often the sensory signals are quantized before being transmitted to the controller. Consequently,  performance is affected by the coarseness of this quantization process. Modern communication technologies allow users to obtain resolution-varying quantized measurements based on the prices paid. 
 In this paper, we consider joint optimal controller synthesis and quantizer scheduling for a partially observed Quantized-Feedback Linear-Quadratic-Gaussian (QF-LQG) system, where the measurements are quantized before being sent to the controller.
  The system is presented with several choices of quantizers, along with the cost of using each quantizer. 
  The objective is to jointly select the quantizers and synthesize the controller to strike an optimal balance between control performance and quantization cost. 
  When the innovation signal is quantized instead of the measurement,  the problem is decoupled into two optimization problems: one for optimal controller synthesis, and the other for optimal quantizer selection.
   The optimal controller is found by solving a Riccati equation and the optimal quantizer selection policy is found by solving a linear program (LP)- both of which can be solved offline. 
    
\end{abstract}



\section{Introduction}
%



Networked control systems operating under finite data-rate constraints employ signal quantization to reduce the amount of data for communication.
System-specific quantizers (encoders) and decoders are designed to compress signals with a finite number of bits and to incur minimal signal reconstruction errors, respectively. 
The available bit-rate to quantize the signals, as well as the choice of the quantizers and the decoders, determine the error in the reconstructed signals, and consequently, they affect the performance of the control system \cite{kostina2019rate, nair2003exponential}.
The quantizers used for networked control systems with limited data-rates are designed to ensure that the least amount of information is lost due to the encoding process.
To achieve this goal, often these quantizers must be time-varying and the dynamics of the time-varying quantizer parameters are tied to the dynamics of the networked control systems for optimal performance \cite{nair2003exponential}.

Time-varying quantizers provide the flexibility to send high resolution quantized signals when needed, and use a coarser resolution otherwise.
Typically, design of dynamic quantizers requires solving a joint optimization problem for the quantizer and the controller \cite{yuksel2013jointly} to obtain optimal performance.
Such co-design the easily becomes intractable due to the nonlinear/saturation behavior of the quantization process. 
Even for the linear-quadratic optimal control problem -- which is one of the simplest problems in optimal control for which an analytical closed-form solution exists -- becomes intractable when quantized measurements are fed back to the controller. 
In \cite{fu2012lack}, the authors show the lack of a \textit{separation principle} for an LQG system with quantized feedback. 
In \cite{yuksel2013jointly}, the authors demonstrated that a separation principle exists when \textit{predictive quantizers} are used. 
Furthermore, this work also demonstrated that the use of predictive quantizers can be made without loss of generality.
The principle behind predictive quantizers is not to quantize the state $X_t$ (which contains all the past control), but rather to quantize a signal that is obtained after removing all the past control  history from $X_t$.
While these works provide some characterization on the optimal quantizer, however, the exact solution of the optimal quantizer is not availabl. 
LQG control with quantization has been studied for a long time \cite{borkar1997lqg, liu1992optimal, tatikonda2000control, tatikonda1998control, tatikonda2004stochastic, williamson1989optimal} and many others. 
Owing to the intractability of the problem, these works do not readily provide the optimal quantizers.
An exception is \cite{bao2010iterative}, where an iterative method is proposed to find a quantizer and a controller for LQG systems.
In principle, the iterative method converges
in the case of open-loop encoder systems. 
However,  as mentioned in that work, the proposed iterative method does not necessarily converge for the general case with partial side information.
For the special case of open-loop encoder systems, the process is likely to converge to a local optimum instead to the global one. 

To circumvent the intractability associated with finding the optimal qunatizers, we formulate a problem to find the optimal quantizer from a \textit{given} finite collection of quantizers for an LQG system. 
In this way, our formulation becomes a quantizer scheduling/selection problem where the best quantizer at each time instance is selected from a given finite set.
We assume a partially observed linear system that can choose from a given set of quantizers to quantize its measurements and transmit the resulting quantized signal to the controller. 
The system can use different quantizers at different time instances to meet the need for time-varying quantizer resolution.
We further assume that these quantizers are costly to use and different quantizers have possibly different costs of operation.
The performance of the system is thus measured by an expected quadratic cost plus the total cost for using the quantizers.
Quantizers with higher resolution are generally more costly than ones with lower resolution.
Therefore, better control performance can be achieved at the expense of a higher quantization cost.
This way, our framework provides a control-quantization trade-off.

Some of the earlier works on quantization and control can be traced backed to   \cite{schweppe1968recursive,curry1970estimation, moroney1983issues, delchamps1989extracting}. 
These works do not necessarily focus on finding the optimal quantizer, but rather investigate the effects of a given quantizer in the system performance.
While optimality (which is measured by a weighted sum of a control and quantization costs) is the focus in the current paper, the works in \cite{delchamps1990stabilizing, wong1997systems,wong1999systems,ishii2003quadratic,li2004robust} consider stability of the system to be their focus.
In \cite{wong1997systems} and \cite{wong1999systems}, the authors explicitly considered the issues of quantization, coding and delay. 
The concept of \textit{containability} was used to study stability  of linear systems.
    A quantization scheme with time-varying quantization sensitivity was studied in \cite{brockett2000quantized} proving asymptotic stability of the system.  
    In \cite{liberzon2003stabilization} the author derived a relationship between the norm of the transition matrix and the number of values taken by the encoder to ensure global asymptotic stability. 
   Reference \cite{nair2003exponential} addressed the problem of finding the smallest data-rate above which exponential stability can be ensured.

In all abovementioned works, 
 the role of quantization has been proven to be  crucial. However, for a given control objective, how to schedule from a set of available quantizers, which have a cost associated with them, has not been addressed.
To the best of our knowledge, \cite{maityTN2021} is the first work where a joint optimization framework is considered to synthesize an optimal controller and schedule the optimal quantizers from a given set of costly quantizers.
This article extends the work of \cite{maityTN2021} by considering a partially observed system with noisy sensors, and, more importantly, it explicitly considers the nuisances of delay and out-of-order measurement availability.

\paragraph{Contributions}

The contributions of this work are as follows.
We show that quantizing the ``innovation signal" separates the controller synthesis problem from the quantizer selection problem. 
While the idea of innovation--quantization was originally proposed in \cite{borkar1997lqg} for a fully observed system with a deterministic initial state, 
in this work, we extend the innovation-quantization idea for partially observed systems with uncertain initial states.  
Furthermore, we explicitly consider delays in the arrival of the measurements at the controller. 
We study the optimal controller and show that the controller is of a certainty-equivalence type. 
The controller gains can be computed offline and they do not depend on the parameters of the quantizers.
 The analysis of the quantizer-selection problem reveals that the optimal strategy for the selection of the quantizers can also be computed offline by solving a simple linear programming problem.

 The rest of the paper is organized as follows: in Section \ref{S:prel} we discuss some background on random variables; in Section \ref{S:prob} we formally define the problem addressed in this paper; Section \ref{S:solution}  provides the structure for the optimal controller and the quantizer selection scheme. 
 Finally, we conclude the paper in Section~\ref{S:conclusion}.
 
\section{Preliminaries} \label{S:prel}

In this section we provide some background on random variables, the Hilbert space of random variables, condition expectation, and the orthogonal projection of random variables defined on a Hilbert space. 

Define the probability space $(\Omega, \sf F, \sf P)$ where $\Omega$ is the sample space, $\sf F$ is the set of events, and the measure $\mathsf{P}: \mathsf{F}\to [0,1]$ defines the probability of occurring an event. 
In this probability space, $X: \Omega \to \X$ is a random variable   defined as a measurable function from the sample space $\Omega$ to a measurable space $\X$, such that for any measurable set $S \subseteq \X$, $X^{-1}(S)=\{\omega\in \Omega: X(\omega) \in S\} \in \sf F$.
 $\E[X]$ denotes the expected value of $X$, with respect to $\mathsf P$, defined as $\E[X]=\int_\Omega X(\omega)\rm{d}\mathsf P(\omega)$.

Let us define the space $\h$ of real-valued ($\X=\R$) random variables $X:\Omega \to \R$  such that
\begin{align*}
\h=\{X|~\E[X^2] < \infty\}.
\end{align*}
For $X,Y \in \h$,  $\alpha X+\beta Y\in \h$ for all $\alpha,\beta \in \R$. 
The inner product in $\h$ is defined by
\begin{align*}
\langle X, Y \rangle =\E[XY].
\end{align*}

\textit{Fact 1} \cite[Section 4.2]{luenberger1997optimization}:  $\h$ is a Hilbert space.\\

Let $X_1,\ldots, X_\ell$ be a collection of $\ell$ random variables belonging to $\h$. 
The $\sigma$-field generated by these random variables is denoted as $\sigma(X_1,\ldots, X_\ell)$, and  the linear span of these random variables is denoted by  $\sigma^L(X_1,\ldots, X_\ell)=\{Y |Y=\sum_{i=1}^\ell c_iX_i, c_i \in \R\}$.
Clearly, we have that $ \sigma^L(X_1,\ldots, X_\ell) \subseteq \sigma(X_1,\ldots,, X_\ell)$. 
The function $g(X_1, \ldots,$ $X_\ell): \R^\ell \to \R$ is a measurable function of the random variables $X_1,\ldots,X_\ell$ if $g^{-1}(S) \in \sigma(X_1,\ldots,X_\ell)$ for all $S\subseteq \R$.
 Let $\mathcal{G}$ denote the set of all measurable functions $g( X_1,\ldots, X_\ell)$ of $\ell$ random variables $X_1,\ldots, X_\ell$. 
 The conditional expectation of a random variable $Y$ conditioned on the random variables $X_1,\ldots, X_\ell$, denoted as  $\E[Y|X_1,\ldots, X_\ell] \in \mathcal{G}$, is defined as \cite[Section 34]{billingsley2008probability} 
 \begin{align*}
 \int_{S} \E[Y|X_1,\ldots, X_\ell]\, \mathrm{d}\mathsf{P}=\int_S Y \mathrm{d}\mathsf{P},~~~~~\forall S \in \sigma(X_1,\ldots,X_\ell).
 \end{align*}

 The following Lemma is adapted from {\cite[Theorem 3.6]{speyer2008stochastic}}.
 
 \medskip
 
\begin{lm} \label{L:ortho} 
For any random variable $Y$, the solution to the optimization problem
\begin{align*}
\stackrel[g\in \mathcal{G}]{}{\inf} \E[(Y-g)^2]
\end{align*}
is $g^*(X_1,\ldots, X_\ell)=\E[Y|X_1,\ldots, X_\ell]$.\\
\end{lm}


That is, $\E[Y|X_1,\ldots, X_\ell]$ is the projection of the random variable $Y$ onto the span of the $\sigma$-field $\mathcal{G}$ generated by $X_1,\ldots, X_\ell$. The projection error $Y-\E[Y|X_1,\ldots, X_\ell]$ is orthogonal to any measurable function $g(X_1,\ldots, X_\ell)\in \mathcal{G}$ (i.e., the error is orthogonal to the $\sigma$-field $\sigma(X_1,\ldots, X_\ell)$),
\begin{align*}
\langle Y-\E[Y|X_1,\ldots, X_\ell], g \rangle =0, \quad \forall g \in \mathcal{G}.
\end{align*}
The following Lemma, presented without proof, states that in the case of Gaussian random variables the conditional expectation can be represented as an affine combination of $X_1,\ldots, X_\ell$.
\begin{lm} \cite[Chapter 11]{davenport1958introduction}
Let $Y,X_1,\ldots, X_\ell$ be jointly Gaussian random variables. Then, there exists $c_0,\ldots, c_\ell \in \R$ such that
\begin{align*}
\E[Y|X_1,\ldots, X_\ell]=c_0+\sum_{i=1}^\ell c_iX_i \in \sigma^L(1,X_1,\ldots, X_\ell).
\end{align*}
\end{lm}
The study in \cite{akyol2012conditions} provides necessary and sufficient conditions for the conditional expectation $\E[Y|X_1,\ldots, X_\ell]$ to be a linear function of $X_1,\ldots, X_\ell$ when the variables are not jointly Gaussian.

The previous definitions and  lemmas can be extended to multi-dimensional random variables \cite{luenberger1997optimization, billingsley2008probability, speyer2008stochastic, davenport1958introduction}. 


 
\section{Problem Formulation} \label{S:prob}

Let us consider a linear discrete-time stochastic system  
\begin{align} \label{E:dyn}
X_{t+1}&=A_tX_t+B_tU_t+W_t,\\
Y_t&=C_tX_t+\nu_t,
\end{align}
where, for all $t\in \bbN_0 ~(= \mathbb{N}\cup\{0\})$, $X_t\in \R^n$, $U_t \in \R^m$ and $Y_t\in \R^p$, $A_t$, $ B_t$ and $C_t$ are matrices of compatible dimensions, $\{W_t\}_{t\in \mathbb{N}_0}$ and $\{\nu_t\}_{t\in \mathbb{N}_0}$ are two i.i.d noise sequences in $\R^n$ and $\R^p$  with statistics $W_0 \sim \N(0,\W)$ and $\nu_0\sim \N(0,\V)$, respectively, and   $W_k$, $\nu_j$ are independent for all $j,k \in \mathbb N_0$. 
The initial state, $X_0$, is also a Gaussian random variable distributed according to $\N(\mu_0,{\Sigma_x})$, and independent of the noises $W_t$ and $\nu_t$ for all $t\in \mathbb{N}_0$. 
For notational convenience, we will write $X_0=\mu_0+ W_{-1}$ where $W_{-1}\sim \N(0,\Sigma_x)$.
Thus, $X_0$, $W_k$, $W_\ell$, $\nu_i$ and $\nu_j$ are independent random variables for all $k,\ell,i,j = 0,1,\ldots$, such that $k\ne \ell$, and $i\ne j$. 
In what follows, we will consider $A_t, B_t$ and $C_t$ to be time invariant in order to maintain notational brevity. 
However, the extension of the results presented in the subsequent sections to time varying $A_t, B_t$ and $C_t$ is straightforward and does not require any further assumptions.

In this work, we address the quantized output feedback LQG (QO-LQG) optimal control problem defined as follows. 
Referring to Figure \ref{Fig:schematic}, we assume that $M$ quantizers are provided to quantize the measurement $Y_t$ and transmit the quantized output to the controller. 
The range of the $i$-th quantizer is denoted by  $\Q^i=\{q_1^i,q_2^i,\cdots,q_{\ell_i}^i\}$. 
Thus, the $i$-th quantizer has $\ell_i$  quantization levels. 
Without any loss of generality, we assume that $\ell_1\le \ldots \le \ell_M$.
Associated with the $i$-th quantizer, let $\p^i=\{\p^i_1,\p^i_2,\cdots,\p^i_{\ell_i}\}$ denote a partition of $\R^p$ such that $\p^i_j$ gets mapped to $q^i_j$ for each $j\in \{1,2,\cdots,\ell_i\}$. 
Specifically, one may think of the $i$-th quantizer as a mapping $g_i:\R^p\to \Q^i$ such that $g_i(y)=q^i_j$ if and only if $y\in \p^i_j$. 


The quantized measurements are transmitted through a  communication channel that has a finite data-rate.
 Consequently, some quantized measurements may need more than one time step to complete the sensor-to-controller transmission and the decoding at the controller's site \cite{arafa2020timely}, and hence, the availability of that measurement to the controller  will be delayed.  
 Furthermore, quantized signals of different lengths may experience  different amounts of delay, and hence,  out-of-order measurement availability is inevitable \cite{kam2013age}. 
 In this work, we do not adhere to any particular model for characterizing this delay, rather we simply consider the case where a quantized signal with larger number of bits may experience a longer delay before it is available to the controller.
  That is, the delay $d_i$ associated with the $i$-th quantizer is non-decreasing with $i$, \eg $d_1\le d_2\le \ldots \le d_M$. 
The number of quantization levels $\ell_i$ generally captures the resolution of the quantization, \eg a higher $\ell_i$ typically means a better resolution and lesser quantization error, but, at the same time, it induces longer delay $d_i$.  
Therefore, this work will also reveal the trade-off between choosing a coarser but faster quantization service versus a finer but delayed service. 
In fact, we will see later on that, for a finite-horizon optimal control problem, different resolution-delay (finer-delayed vs$.$ coarser-faster) characteristics are preferred at different time instances.

Associated with each quantizer there is an operating cost that must be paid in order to use this quantizer. 
Let $\lambda(\Q^i)=\lambda_i\in \R_+$ denote the cost associated with the $i$-th quantizer. 
For example, $\lambda_i \propto \log_2\ell_i$ represents the case where the 
cost is proportional to the code-length used to encode the output of the quantizer. 
This cost is also related to the delay associated with the controller.   
In this work, we do not adhere to any specific structure for $\lambda$.
 We just assume that the values of $\lambda_i$'s are given to us a priori. 
 If there is a cost for operating the communication channel, that cost can be also incorporated into $\lambda_i$. 


Note that, in contrast to previous works~\cite{ wong1997systems, elia2001stabilization}, 
we do not aim at designing a quantization scheme, rather a set of quantizers is already given by some service provider. 
Our objective is to optimally decide  which quantizer is to be requested for use at what time instances.
Also, we will assume that the costs $\lambda_i$ are determined by the service provider and presented to us a priori. 
Designing such costs in order to regulate the use of the quantizers is an equally interesting problem for the service provider that will be addressed elsewhere. 
We will further assume that the communication channel  between each quantizer and the controller always transmits the quantized information without any distortion.

\begin{figure}
\centering
\includegraphics[draft=false, width=0.7 \textwidth]{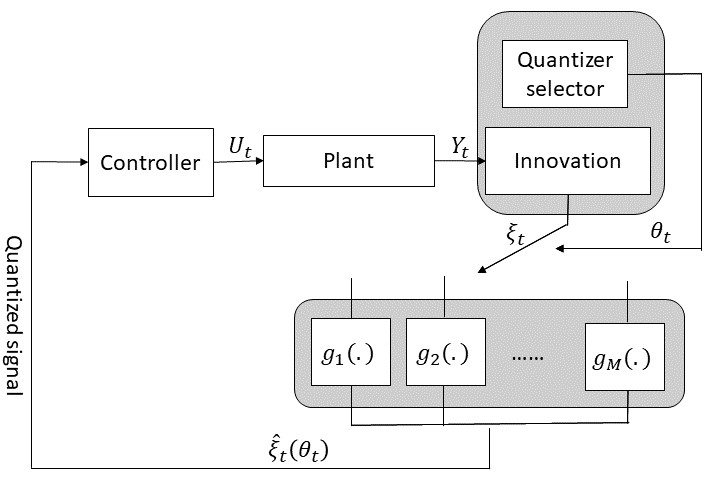}
\caption{Schematic diagram of the system. The top-right gray block contains the quantizer selector that selects the optimal quantizer at each time, and the innovation block that produces the innovation signals from the measurements. The down-right gray block contains the set of $M$ quantizers whose outputs are sent through the communication channel to the controller.} \label{Fig:schematic}
\vspace*{-10 pt}
\end{figure}

The objective is to minimize a performance index that takes into account the quantization cost. 
Contrary to the existing literature  on quantization-based LQG \cite{borkar1997lqg,tatikonda2000control, tatikonda1998control, tatikonda2004stochastic, williamson1989optimal, liu1992optimal}, in our case there are two decision makers instead of a single one: One decision-maker (the controller) decides the input ($\{U_t\}_{t\in \bbN_0}$) to apply to the system, and the other decision-maker (the quantizer-selector) decides the quality and delay of the measurements (quantized state values) which are transmitted to the controller.

 We introduce a new decision variable $\theta^i_t$ for the quantizer-selector in the following way:
\begin{align*}
\theta_t^i =\begin{cases} 1, ~~~&i\text{-th quantizer is used at time } t,\\
 0,~~~&\text{otherwise.}
\end{cases}
\end{align*}
Let us denote the vector $\theta_t\triangleq[\theta_t^1,\theta_t^2,\ldots,\theta_t^M]\T \in \{0,1\}^M$, that characterizes the decision of the quantizer-selector at time $t$. 
We enforce the quantizer-selector to select only one quantizer at any time instance, and hence   for all $t\in \bbN_0$, we have 
\begin{align}\label{E:ThetaConstraint}
\sum_{i=1}^M\theta_t^i=1.
\end{align}

The decoded measurement(s) available to the controller at time $t$ is denoted as $\hat O_t$. 
Note that $\hat O_t$ may contain delayed quantized measurements; also, several measurements may be made available simultaneously at the controller. 
For example, as shown in Figure \ref{F:out}, if there are two quantizers with $d_1=1$ and $ d_2=3$, and if the second quantizer is selected at time $0$ followed by the selection of the first quantizer at times $t=1,2$, then no decoded measurements are available at times $t=0,1$, \eg $\hat O_0=\hat O_1=\emptyset$.
The decoded information about $Y_1$, denoted as $\hat Y_1$, is available at time $t=2$, \eg $\hat{O}_2=\{\hat{Y}_1\}$, and the decoded information about $Y_0$ and $Y_2$ are available simultaneously at time $t=3$, \eg $\hat{O}_3=\{\hat{Y}_0,\hat{Y}_2\}$.
Thus, $\hat O_t$ is a function of $\{\theta_0,\ldots,\theta_t\}$ (to be precise, $\hat O_t$ is only a function of $\{\theta_{t-d_i}: ~i=1,\ldots, M,~ t-d_i\ge 0 \}$).
\begin{figure}[h]
\begin{center}
\begin{tikzpicture}
\draw[->, gray, thick] (0,0) -- (7.5,0) coordinate (time);
\foreach \x in {0,...,7}
\draw (\x cm,-1pt) -- (\x cm,1pt) node[below] {$\x$};
\node at (0,-.5) {$\hat O_0=\emptyset$};
\node at (1,-1) { $\hat O_1=\emptyset$};
\node at (2,-.5) { $\hat O_2=\{{\color{red}\hat Y_1}\}$};
\node at (3,-1) { $\hat O_3=\{{\color{teal}\hat Y_0},{\color{blue} \hat Y_2}\}$};
\node at (4,-.5) { $\hat O_4=\emptyset$};
\draw[->, color=teal] (0,0)  to [in=90, out =90] (3,0) node[anchor=south east] {$\hat{Y}_0$};
\draw[->, color=red!80] (1,0)  to [in=120, out =60] (2,0) ;
\node[color=red!80] at (2.15, 0.2) {$\hat{Y}_1$};
\draw[->, color=blue] (2,0)  to [in=120, out =60] (3,0) ;
\node[color=blue] at (3.25, 0.3) {$\hat{Y}_2$};
\draw[->, color=black] (3,0)  to [in=120, out =60] (6,0) ;
\draw[->, color=purple] (4,0)  to [in=120, out =60] (7,0) ;
\draw[->, color=brown] (5,0)  to [in=120, out =60] (6,0) ;
\end{tikzpicture}
\caption{Out-of-order measurement availability at the controller when the second quantizer (with delay $3$) is selected at times $t=0,3,4$ and the first quantizer (with delay $1$) is selected at other time instances. 
The new decoded measurements available at time $t$ at the controller is $\hat O_t$, \eg $\hat O_0=\hat O_1=\emptyset,\hat O_2= \{\hat Y_1\}, \hat O_2=\{\hat Y_0,\hat Y_2\}$, and so on. 
In this example, $\hat{Y}_1$ is available before $\hat{Y}_0$ and $\hat{Y}_5$ is available before $\hat{Y}_4$.} \label{F:out}
\end{center}
\vspace*{-10 pt}
\end{figure}
A detailed description of $\hat O_t$ will be provided later on.

Let us introduce the sets  $\Y_t\triangleq \{Y_0,Y_1,\cdots,Y_t\}$, $\hat{\mO}_t \triangleq \{\hat O_0,\hat O_1,\cdots,\hat O_t\}$ $\U_t\triangleq\{U_0,U_1,\cdots,U_t\}$ and $\Theta_t\triangleq\{\theta_0,\theta_1,\cdots,\theta_t\}$ to be the measurement history, quantized measurement history at the controller, control history, and quantization-selection history, respectively. 
For convenience, we will use the notation $\U$ for $\U_{T-1}$, and likewise, we will use $\Theta$ for $\Theta_{T-1}$.

The information available to the controller at time $t$ is $\I_t^c=\{\hat \mO_t,\U_{t-1}\}=\I_{t-1}^c\cup\{\hat O_t,U_{t-1}\}$ where $\I_0^c=\{\hat O_0\}$. 
It should be noted that $\I^c_t$ depends on $\Theta_t$ through $\hat\mO_t$. 
In classical optimal LQG control, the information available to the controller is not decided by any active decision maker, unlike the situation here.
An admissible control strategy at time $t$ is a measurable function from the Borel $\sigma$-field generated by $\I_t^c$ to $\R^m$. Let us denote such strategies by $\gamma^u_t(\cdot)$ and the space they belong to by $\Gamma^u_t$.
On the other hand, the information available to the quantizer-selector at time $t$ is $\I_t^q=\{\Y_{t},\hat \mO_{t-1},\U_{t-1},\Theta_{t-1}\}=\I_{t-1}^q\cup\{Y_t,\hat O_{t-1},U_{t-1},\theta_{t-1}\}$ where $\I_0^q=\{Y_0\}$.
 The information $\I^q_t$ will be used to generate a signal $\xi_t=f(\I^q_t)$ that will  further be quantized before being transmitted to the controller. 
If $f(\I^q_t)=Y_t$, then the output itself is quantized. 
The information $\bar{\I}^q_t=\{\hat \mO_{t-1},\Theta_{t-1}\} \subset \I^q_t$ will be used to decide the optimal quantizer to quantize $\xi_t$.
Thus, the admissible strategies for the selection of the quantizers are measurable functions from the Borel $\sigma$-field generated by $\bar \I_t^q$ to $\{0,1\}^M$. Let us denote such strategies by $\gamma^\theta_t(\cdot)$, and the space they belong to by $\Gamma^\theta_t$.
Thus, the entire quantization process is characterized by the following two equations:
\begin{subequations} \label{E:quantization}
\begin{align}
\xi_t=&f(\I^q_t),\\
\theta_t=&\gamma^\theta_t(\bar \I_t^q).
\end{align}
\end{subequations}

For brevity, we will often use $\gamma^u_t$ instead of $\gamma^u_t(\cdot)$ or $\gamma^u_t(\I_t^c)$, and $\gamma^\theta_t$ in place of $\gamma^\theta_t(\cdot)$ or $\gamma^\theta_t(\bar\I^q_t)$.
 Let $\gamma^\Theta$ denote the entire sequence $\{\gamma^\theta_0,\gamma^\theta_1,\cdots,\gamma^\theta_{T-1}\}$ and let $\Gamma^\Theta$ denote the space where $\gamma^\Theta$ belongs to. 
 Likewise, $\gamma^\U$ and $\Gamma^\U$ are defined similarly. Let us also define $\I^c=\{\I^c_t\}_{t=0}^{T-1}$ and $\I^q=\{\I^q_t\}_{t=0}^{T-1}$.
The sequence of decision making within one time instance is then as follows: \\
{{$$\I_t^q\overset{f,\gamma^\theta_t}{\rightarrow}  \{\xi_t,\theta_t\}\to \hat O_t \to \I_t^c\overset{\gamma^u_t}{\to} U_t\to X_{t+1}\to Y_{t+1}\to \I_{t+1}^q.$$}}


The cost function to be minimized cooperatively by the quantizer-selector and the controller is a finite horizon expected quadratic criterion,  given as
\begin{align} \label{E:cost}
J(\U,\Theta)=\E\left[\sum_{t=0}^{T-1}(X_t\T Q_1X_t+U_t\T RU_t+\theta_t\T \Lambda) +X_T\T Q_2X_T \right],
\end{align}
where $\Lambda=[\lambda_1,\lambda_2,\ldots,\lambda_M]\T $ is the cost for quantization, $Q_1,Q_2 \succeq 0$, $R\succ 0$, $\U=\gamma^\U(\I^c)=\{\gamma^u_0(\I_0^c), \gamma^u_1(\I_1^c),$ $\ldots,\gamma_{T-1}^u(\I_{T-1}^c)\}$ and  $\Theta=\gamma^\Theta(\bar\I^q)=\{\gamma^\theta_0(\bar\I_0^q), \gamma^\theta_1(\bar\I_1^q),$   $\ldots,\gamma_{T-1}^\theta(\bar\I_{T-1}^q)\}$.
We seek to find the optimal strategies $\gamma^{\U*}=\{\gamma_0^{u*},\gamma_1^{u*}, $ $\ldots,\gamma_{T-1}^{u*}\}$ and $\gamma^{\Theta*}=\{\gamma^{\theta*}_0,\gamma^{\theta *}_1,\ldots,\gamma^{\theta *}_{T-1}\}$ that minimize \eqref{E:cost}. We will also rewrite \eqref{E:cost} in terms of $\gamma^\U$ and $\gamma^\Theta$ as
\begin{align} \label{E:cost2}
J(\gamma^\U,\gamma^\Theta)=\E\Big[\sum_{t=0}^{T-1}(X_t\T Q_1X_t+U_t\T RU_t+&\theta_t\T \Lambda) +X_T\T Q_2X_T\nonumber\\&|~U_t=\gamma^u_t(\I^c_t), \theta_t=\gamma^\theta_t(\bar\I^q_t)\Big].
\end{align}

The cost function \eqref{E:cost2} is affected by the choice of the function $f(\I^q_t)$. 
Solving an estimation problem is intractable even when $\xi_t=f(\I^q_t)=Y_t$ and there is only one quantizer, let alone the control problem with multiple quantizers; for example, confer \cite{duan2008state, clements1972approximate, karlsson2005filtering, sviestins2000optimal} and the references therein.
 Although a linear quadratic Gaussian system is considered here, the non-linearity associated with the quantization process 
 makes the problem challenging, since quantization results in a nonlinear stochastic optimal control problem. 
To keep our analysis tractable, in this paper, we will consider 
$$
\xi_t=f(\I^q_t)=Y_t-\E[Y_t|Y_0,\ldots,Y_{t-1}],
$$ 
that is, the innovation signal.
 Quantizing the innovation signal not only makes the problem tractable, but also allows us to show that a separation principle between control and quantizer-selection is retained. 
It is well known \cite{kailath1968innovations} that the information contained in the innovation signals $\{\xi_0,\ldots,\xi_t\}$ is the same as the information contained in the observations $\{Y_0,\ldots,Y_t\}$. 
Therefore, designing an output-feedback controller is equivalent to designing an innovation-feedback controller.
 However, after quantization, the information contained in the quantized innovations is not necessarily the same as the information contained in the quantized outputs. 
 Therefore, in general, it cannot be claimed that the performance of the optimal output-quantized feedback  controller will be the same as that of the optimal innovation-quantized feedback  controller.

In the following, the information $\I_t^q=\{\Y_{t},\hat \mO_{t-1}, \U_{t-1}, \Theta_{t-1}\}$ will be divided into two parts, namely, $\{\Y_{t},\U_{t-1}\}$, which will be used for generating the innovation signals $\xi_t$, and $\bar\I_t^q=\{\hat \mO_{t-1},\Theta_{t-1}\}$, which will be used for selecting the quantizers. 
Therefore, \eqref{E:quantization} takes the form
\begin{subequations} \label{E:quant}
\begin{align}
\xi_t=&Y_t-\E[Y_t|\Y_{t-1},\U_{t-1}],  \label{E:quantA}\\
\theta_t=&\gamma^\theta_t(\bar\I_t^q) = \gamma^\theta_t(\{ \hat \mO_{t-1},\Theta_{t-1}\}).
\end{align}
\end{subequations}
At this point, one may notice that the presence of $\U_{t-1}$ is redundant in \eqref{E:quantA} since $U_t$ is a function of $\I^c_t$ which can be written as some function (that depends on $\gamma^u_t,\gamma^\theta_t,$ and $f$) of $\Y_t$.

\section{Optimal Control and Quantization Selection} \label{S:solution}

In this section we find the optimal $\gamma^{\U*}$ and $\gamma^{\Theta *}$ that minimize the cost function \eqref{E:cost2} amongst all admissible strategies, that is,
\begin{align}\label{E:argmin}
(\gamma^{\U*},\gamma^{\Theta*})=\underset{\gamma^\U\in \Gamma^\U,\gamma^\Theta\in \Gamma^\Theta}{\arg\min}J(\gamma^\U,\gamma^\Theta).
\end{align}

Before proceeding further to solve \eqref{E:argmin}, let us discuss, in some detail, the input for the quantization process since it will play a crucial role in the following analysis. Unlike other quantized feedback-based control approaches \cite{williamson1989optimal}, \cite{liu1992optimal}, we will quantize an innovation signal $\xi_{t}$ instead of $Y_t$ at time $t$. 
 The innovation signal $\xi_{t}$ can be readily computed from the measurement history $\Y_t$ as follows. 
 Let $\h$ be the Hilbert space of  random variables in $\R^p$ having finite covariances. 
 The observations $Y_0,Y_1,\ldots,Y_t$ belong to $\h$, and the $\sigma$-field generated by these random variables is denoted by $\sigma(\Y_t)=\sigma(Y_0,\ldots,Y_t)$. 
 With a slight abuse of notation, we will use $\Y_t$ to denote both the $\sigma$-field $\sigma(\Y_t)$ and the set of random variables $\{Y_0,Y_1,\ldots,Y_t\}$, whenever the context is not ambiguous. 
 These random variables may not necessarily be orthogonal, \eg $\E[Y_iY_j\T]\ne 0$. However, one can construct random variables $\xi_0,\xi_1,\ldots,\xi_t$ which are orthogonal and $\sigma(\xi_0,\ldots,\xi_t)=\sigma(\Y_t)$. 
 It can be shown that the random variable $\xi_i$ is of the form $\xi_i=Y_i-\E[Y_i~|\Y_{i-1}]$; see  \cite{kailath1968innovations}.
  In order to prove the orthogonality of $\xi_i,\xi_j$, let us consider $i> j$ (hence $\Y_{i-1}\supseteq \Y_j$), and observe that
 \begin{align*}
 \E[\xi_i \xi_j\T] &=\E\left[\E[\xi_i \xi_j\T~|\Y_j]\right] \\
 &=\E\left[\left(\E[\xi_i~|\Y_j]\right) \xi_j\T\right] \\
 &=\E\left[\left(\E[Y_i-\E[Y_i~|\Y_{i-1}]~|\Y_j]\right) \xi_j\T\right] \\
 &=\E[0\xi_j\T]=0.
 \end{align*}

    \subsection{The Innovation Process}

The control $U_t$ is a function of the quantized innovations which are not Gaussian random variables.
 Therefore, the state $X_t$ and the measurement $Y_t$ are no-longer Gaussian random variables under quantized innovation feedback. 
Although the innovation signal is a Gaussian random variable for partially observed classical linear-quadratic-Gaussian systems without quantization, in our case, this may no longer be true since the control is a function of quantized signals 
(which are not Gaussian random variables).
We therefore need to independently verify whether the distribution of the innovation signal is Gaussian or not. 

It can be verified that  the innovation $\xi_t$ is not affected by the control strategy, although, $Y_t$ is affected. 
 Furthermore, the innovation $\xi_t$ retains its Gaussian distribution and the parameters of this distribution can be computed offline. 
 This observation is presented in the following proposition.

%

\medskip

\begin{pr} \label{Pr:gaussian}
For all $t$, $\xi_t$ is a Gaussian random variable with zero mean and covariance $M_t$ such that
\begin{align*}
M_{t+1}&=C\Sigma_{t+1|t} C\T+\V\\
\Sigma_{t+1|t}&=A\Sigma_t A\T +\W, \quad \Sigma_{0|-1} = \Sigma_x\\
\Sigma_{t+1}&=\Sigma_{t+1|t}-\Sigma_{t+1|t}C\T M_{t+1}^{-1}C \Sigma_{t+1|t}.
\end{align*}
Moreover, the sequence of random variables $\{\xi_0,\ldots,\xi_t\}$ is uncorrelated for all $t$.
\end{pr}

\begin{proof}
The proof is presented in Appendix~\ref{AP:1}.
\end{proof}
\vspace{6 pt}  
 Proposition \ref{Pr:gaussian} is equivalent of the following facts:
  \begin{enumerate}
\item The innovation sequence $\{\xi_t\}_{t \in \bbN_0}$ does not depend on the control history $\U_{t-1}$.
\item The innovation sequence is a Gaussian uncorrelated noise sequence with zero mean and covariance $M_t$.
\item Since the sequence of random variables $\{\xi_t\}_{t \in \bbN_0}$ is uncorrelated and Gaussian,  each $\xi_t$ and $\xi_k$ are independent for all $k\ne t$.
\end{enumerate}  
  
 
   \subsection{Implications of Delay}
 
 Let $g_i(\xi_{t}) \in \Q^i$ denote the quantized version of  $\xi_{t}$ if the $i$-th quantizer is selected. 
 Therefore, the  quantized information sent to the controller is
 \begin{align} \label{E:measure}
\hat \xi_t=\sum_{i=1}^Mg_i(\xi_{t})\theta^i_t,
 \end{align}
and this information will be decoded and be available at the controller at time $t+\sum_{i=1}^M\theta^i_td_i$.
  Notice that $g_i(\xi_t)\in \Q^i$ is a random variable, and hence  $\hat \xi_t$ is a random variable taking values in the discrete set $\cup_{i=1}^M\Q^i$ with $\mathsf{P}(\hat \xi_t=q^i_j)=\mathsf{P}( \xi_{t}\in \p^i_j)$. 
 
Since the delays may result in out-of-order availability of the decoded signal to the controller, it is important that every quantized signal is time-stamped, i.e., when the controller receives a decoded measurement $\hat q$ at time $t$, it should be able to uniquely determine which of the signals $\{\xi_0,\ldots,\xi_t\}$ was quantized to produce this measurement along with the quantizer that was used. 
In order to uniquely decode which of the signals $\{\xi_0,\ldots,\xi_t\}$ produced the data $\hat q$, the pair $(\hat{\xi}_t,i)$ will be sent at each time $t$, where $i$ is the index of the quantizer that was used to quantize $\xi_t$.  
Consequently, if the decoded pair $(\hat q,i)$ is received by the controller at time $t$, then the controller can immediately  infer that the $i$-th quantizer was used and that this signal is delayed by $d_i$ units, and hence $\hat q$ corresponds to $\xi_{t-d_i}$.
 Thus, $(\hat q,i)$ reveals that $\theta^i_{t-d_i}=1$, and $\hat q=g_i(\xi_{t-d_i})$.
 At any time $t$, there can be at most $M$ (delayed) new simultaneously available decoded measurements.
 Let us define the set of indexes that are present in $\hat{O}_t$ by
 \begin{align*}
 {\rm idx}_t=\{i: \exists q \in \R^p \textrm{  s.t.  }  (q,i) \in \hat O_t\} \subseteq \{1,\ldots,M\}.
\end{align*}  
 
 Therefore, $\theta_{t-d_i}^i=1$ if $i \in \textrm{idx}_t$, otherwise $\theta_{t-d_i}^i=0$. 
 It follows that the new decoded measurements available to the controller at time $t$ can be expressed as:
 \begin{align*}
 \{\theta_{t-d_1}^1,\ldots,\theta_{t-d_M}^M\}\cup\{\hat{\xi}_{t-d_i}: i \in \mathrm{idx}_t\}.
 \end{align*}
 With a slight abuse of notation, the above set is equivalent to:
 \begin{align*}
 \{\theta_{t-d_1}^1,\ldots,\theta_{t-d_M}^M,\theta_{t-d_1}^1\hat{\xi}_{t-d_1},\cdots, \theta_{t-d_M}^M\hat{\xi}_{t-d_M}\}.
 \end{align*}
Having characterized the effects of delays in the information available to the controller, in the next section, we discuss the optimal controller that minimizes the cost function \eqref{E:cost2}.

\subsection{Optimal Control Policy}

 Let us define the innovation history by  $\Xi_t\triangleq\{\xi_0,\ldots,\xi_t\}$. 
 With a slight abuse of notation, we also denote  $\Xi_t = \sigma(\xi_0,\ldots,\xi_t)$ to be the $\sigma$-field that generated these innovation signals.
 Let us further define the state estimate by
\begin{align}
\bar{X}_t\triangleq \E[X_t|\I^c_t].
\end{align} 
 The quantized information available to the controller at time $t$ is $\hat{\mO}_t=\{\vartheta_{0,t}\hat\xi_0,\vartheta_{1,t}\hat{\xi}_1, $ $\cdots,\vartheta_{t,t}\hat{\xi}_t\}\cup_{k=0}^t\{\theta_{k-d_i}^i: i=1,\ldots,M, k-d_i \ge 0\}$, where $\vartheta_{k,t}$ is an indicator of whether $\hat\xi_k$ is available at the controller by time $t$ or not.
 Note that  $\vartheta_{k,t}$ can be expressed as
 \begin{align}\label{E:vartheta}
 \vartheta_{k,t}=\sum_{i=0}^M\theta_{k}^i 1_{d_i\le t-k}.
 \end{align}

 Clearly, if $t-k\ge d_M$ for some $k$, then the above expression for $\vartheta_{k,t}$ becomes $\vartheta_{k,t}=\sum_{i=0}^M\theta_{k}^i=1$ ensuring that the quantized version of $\xi_k$ is present at the controller. 
  
 Similarly to $\hat{\mO}_t$, let us define the set $\Z_t=\{\vartheta_{0,t}\xi_0,\vartheta_{1,t}{\xi}_1,\cdots,\vartheta_{t,t}{\xi}_t\}\cup_{k=0}^t\{\theta_{k-d_i}^i: i=1,\ldots,M, k\ge d_i \}$, which contains the innovation signals that were quantized to produce $\hat{\mO}_t$ along with the corresponding indexes of the quantizers that were used. 
  Due to the construction of $\Z_t$, $\hat \mO_t$ does not contain any new information when $\Z_t$ is given\footnote{%
  $\mO_t$ contains the innovation signals that were quantized to produce $\hat{\mO}_t$ as well as the 
  corresponding indices of the quantizers that were used for quantization.
  Therefore, each $\hat\xi_k\in \hat \mO_t$ can be computed from $\mO_t$.}. 
 Therefore, we have
 \begin{align} \label{eq:bar_X_iterated_expectation}
 \begin{split}
     \bar X_t&=\E[X_t|\I^c_t]=\E[X_t|\hat \mO_t,\U_{t-1}]=\E[\E[X_t|\Z_t,\hat \mO_t,\U_{t-1}]|\hat \mO_t,\U_{t-1}] \\
     &=\E[\E[X_t|\Z_t,\U_{t-1}]|\hat \mO_t,\U_{t-1}]
 \end{split}
 \end{align}
 In order to compute $\bar X_t$, we compute $\E[X_t|\Z_t,\U_{t-1}]$ which is inside the outer expectation of the last equation.
 
\begin{lm} \label{L:inno}
For any $t$, 
\begin{align}
\E[X_t|\Z_t,\U_{t-1}]=A^t\mu_0+\sum_{k=0}^t\Psi(t,k)\vartheta_{k,t}\xi_k+\sum_{k=0}^{t-1}A^{t-1-k}BU_k,
\end{align}
and, for all $t\ge k$, the matrices $\Psi(t,k)$ are given by
\begin{align} \label{E:psi}
\Psi(t,k)=A^{t-k}\Sigma_{k|k-1}C\T M_k^{-1}.
\end{align}
\end{lm} 
 
\begin{proof}
The proof is given in Appendix~\ref{AP:inno}.
\end{proof} 
 
 Therefore, using Lemma~\ref{L:inno} we obtain from \eqref{eq:bar_X_iterated_expectation} that
 \begin{align} \label{E:x_bar_new}
 \begin{split}
     \bar X_t =& \E[\E[X_t|\Z_t,\U_{t-1}]|\hat \mO_t,\U_{t-1}] \\
     =&A^t\mu_0+\sum_{k=0}^t\Psi(t,k)\vartheta_{k,t} \E[\xi_k\mid \hat{\mO}_t]+\sum_{k=0}^{t-1}A^{t-1-k}BU_k,
     \end{split}
 \end{align}
 where, we have used the fact that $U_{t}$ is a measurable function of $\I^c_t=\{\hat{\mO}_t,\U_{t-1}\}$ for all $t$, and hence, given $\hat{\mO}_t$, the control history $\U_{t-1}$ does not provide any new information about $\xi_k$, i.e., $\E[\xi_k\mid \hat{\mO}_t,\U_{t-1}]= \E[\xi_k\mid \hat{\mO}_t]$.
 Next, we focus on computing $\E[\xi_k\mid \hat{\mO}_t]$.
 To that end, let us define $\bar{\xi}_t^i\triangleq \E[\xi_t|\hat\xi_t,\theta_t^i=1]$. Based on \eqref{E:ThetaConstraint} and \eqref{E:measure}, we may write 
\begin{align}
\bar{\xi}_t^i&= \E[\xi_t|g_i(\xi_t),\theta_t^i=1] \nonumber \\
&=\sum_{j=1}^{\ell_i}1_{g_i(\xi_t)=q^i_j}\E[\xi_t|g_i(\xi_t)=q^i_j,\theta_t^i=1]\nonumber \\
&=\sum_{j=1}^{\ell_i}1_{g_i(\xi_t)=q^i_j}\E[\xi_t|\xi_t\in \p^i_j] \nonumber\\
&=\sum_{j=1}^{\ell_i}1_{g_i(\xi_t)=q^i_j}\int_{\p^i_j}\xi\mathsf{P}_t(\mathrm d\xi|\p^i_j)
\end{align}
where $1_{a=b}$ is an indicator function that is equal to $1$ if and only if $a=b$, otherwise it equals  $0$. 
Therefore, $\bar\xi^i_t$ is a random variable taking values in the set $\{\int_{\p^i_j}\xi\mathsf{P}_t(\mathrm d\xi|\p^i_j): j=1,\ldots,\ell_i\}$ and it depends on the realization of $\xi_t$ through $1_{g_i(\xi_t)=q^i_j}$.
From Proposition \ref{Pr:gaussian}, one can easily compute  the measure $\mathsf{P}_t(\mathrm d\xi|\p^i_j)$ as follows
\begin{align}
\mathsf{P}_{t}(\mathrm d\xi|\p^i_j)&=\begin{cases} \alpha_t e^{-{\xi\T M_{t}^{-1}\xi}/{2}}\mathrm d\xi,
 ~~~~~& \xi \in \p^i_j,\\
0, ~~~&\text{otherwise},\end{cases}\\
(\alpha_t)^{-1}&=\sqrt{(2\pi)^p\det(M_t)}\mathsf{P}(\xi_t\in \p^i_j), \nonumber \\
&=\int_{\p^i_j}e^{-{\xi\T M_{t}^{-1}\xi}/{2}}\mathrm d\xi.
\end{align}
Furthermore, from Proposition \ref{Pr:gaussian}, we have that $\xi_t\sim\N(0,M_t)$. 
Since $M_t$ can be computed offline,  the prior distribution of $\xi_t$ is known to the controller. After receiving the quantized value $\hat\xi_t$, the controller updates the distribution of $\xi_t$. 
 If the quantized value of $\xi_t$, after being quantized by the $i$-th quantizer, is $\hat\xi_t=q^i_j$, then the controller can infer that  $\xi_t \in \p^i_j$.
  This is illustrated in Figure \ref{F:illus}.
  
\begin{figure}
\centering
\begin{tikzpicture}
\draw[<->] (0,0) -- (4,0) ;
\draw[->] (2,0) node[below] {$\underbrace{~~~~~~~~}_{\p^i_2}$} --(2,3) ;
\draw[<->]  (5,0) -- (9,0);
\draw[->] (7,0) --(7,3) ;
\draw[gray!80, step =.5 cm] (0.1, 0) grid (3.9, 2.9);
\draw[gray!80, step =.5 cm] (5.1, 0) grid (8.9, 2.9);

\def\fy{2.9*1/exp(((\y-7)^2)/1)}

\def\normaltwo{\x,{2.5*1/exp(((\x-2)^2)/1)}}
\def\normaltwoo{\x,{2.9*1/exp(((\x-7)^2)/1)}}

\draw[color=blue,domain=0.2:3.8] plot (\normaltwo) node[right] {};
\draw[color=blue,domain=6.5:7.5] plot (\normaltwoo) node[right] {};
\fill[color = orange, opacity =0.2] (1.5,0) rectangle (2.5,3);

\node at (3.25, -0.45) {$\underbrace{~~~~~~~~}_{\p^i_3}$};

\node at (0.75, -0.45) {$\underbrace{~~~~~~~~}_{\p^i_1}$};
\foreach \y in {6.5,7.5}
\draw[color=blue] ({\y},{\fy}) -- ({\y},0);
\draw[color=blue] (7.5,0) -- (9,0);
\draw[color=blue] (5,0)-- (6.5,0);
\node at (2,-1) {(a)};
\node at (7,-1) {(b)};

\end{tikzpicture}
\caption{(a): The blue curve denotes the prior distribution {$\mathsf{P}_t(\mathrm d\xi)$}. The partitions $\p^i_j$ for the $i$-th quantizer is shown as well where $\p^i_2$ is highlighted with the orange block. (b) The posterior distribution ({$\mathsf{P}_t(\mathrm d\xi|\p^i_2)$}) of $\xi_t$ is shown here for the case when the received quantized measurement $\hat\xi_t$ is $q^i_2$ (or equivalently, $\xi_t \in \p^i_2$).} \label{F:illus}
\end{figure}


  The entity $\bar{\xi}^i_t$ computes the expected value of $\xi_t$ given that the $i$-th quantizer was used in the process of quantization, and the quantized value is $\hat\xi_t \in \Q^i$.
   Now, let us further define
\begin{align} \label{E:xi_bar}
\bar{\xi}_t\triangleq \E[\xi_t|\hat{\xi}_t,\theta_t]=\sum_{i=1}^M\theta^i_t\bar{\xi}^i_t,
\end{align}   
and  
\begin{align} \label{E:xi_tilde}
\tilde\xi_t\triangleq \xi_t-\bar\xi_t.
\end{align}   
   From this definition of $\bar\xi_t$, along with the constraint $\sum_{i=1}^M\theta^i_t=1$, we have that $\bar\xi_t=\bar\xi_t^i$ if and only if the $i$-th quantizer was selected at time $t$.
    The conditional covariance $\mathcal{M}_t(\theta_t) \triangleq \E[\tilde{\xi}_t\tilde{\xi}_t\T \mid \theta_t]$ turns out to be  
    \begin{align} \label{E:mcal}
\mathcal{M}_t(\theta_t)=&\E\left[\xi_t\xi_t\T-\xi_t{\bar{\xi}_t}\T -\bar{\xi}_t\xi_t\T+\bar{\xi}_t\bar{\xi}_t\T \mid \theta_t \right] \nonumber \\
=&\E[\xi_t\xi_t\T\mid \theta_t]-\E[\bar{\xi}_t\bar{\xi}_t\T \mid \theta_t] 
= \E[\xi_t\xi_t\T\mid \theta_t] - \E[\bar{\xi}_t\bar{\xi}_t\T \mid \theta_t] ,
\end{align}
where we have used the fact that $\E[\xi_t\bar{\xi}_t\T \mid \theta_t ]=\E[\E[\xi_t\bar{\xi}_t\T \mid  \hat{\xi}_t,\theta_t]\mid \theta_t ] = \E[\E[\xi_t|\hat{\xi}_t,\theta_t]\bar{\xi}_t\T \mid \theta_t] = \E[\bar{\xi}_t\bar{\xi}_t\T \mid \theta_t]$. 
By defining $F_t(\theta_t) \triangleq\E[\bar{\xi}_t\bar{\xi}_t\T \mid \theta_t]$ and using the expression of $\bar{\xi}_t$ from \eqref{E:xi_bar}, we obtain
   \begin{align} \label{E:F}
   F_t(\theta_t)&=\E[\bar{\xi}_t\bar{\xi}_t\T\mid \theta_t] = \sum_{i=1}^M\theta^i_t\E[\bar{\xi}_t^i\bar{\xi}_t^{i\T}]=\sum_{i=1}^M\theta^i_t F^i_t,
   \end{align}
   where
   \begin{align} \label{E:Fi}
   F^i_t&=\E[\bar{\xi}_t^i\bar{\xi}_t^{i\T}]=\sum_{j=1}^{\ell_i}\mathsf{P}(\xi_t\in \p^i_j)\E[\xi_t|\xi_t\in \p^i_j]\E[\xi_t|\xi_t\in \p^i_j]\T.
   \end{align}
   Therefore, using the definition of $F_t(\theta_t)$, we may rewrite \eqref{E:mcal} as $\mathcal M_t(\theta_t) = \E[\xi_t\xi_t\T\mid \theta_t] - F_t(\theta_t)$, and furthermore, we also obtain $\E[{\mathcal{M}_t(\theta_t)}]=M_t - \E[F_t(\theta_t)]$.
   The linear dependence of $F_t(\theta_t)$ on $\theta_t$ will be useful in designing a linear-program for selecting the optimal quantizers, as shown later in the paper.
 
 At this point,  recall from Proposition \ref{Pr:gaussian} and the discussion thereafter that $\{\xi_t\}_{t\in \bbN_0}$ is a sequence of uncorrelated zero-mean Gaussian noises (hence $\xi_k,\xi_\ell$ are independent for $k\ne \ell$) and $\{\hat\xi_t\}_{t\in \bbN_0}$ is the corresponding sequence of the quantized version of $\{\xi_t\}_{t\in \bbN_0}$. 
 Therefore, $\xi_k$ and $\hat{\xi}_\ell$ are independent for all $k\ne \ell$. 
 Hence,
 \begin{align} \label{E:xi_expected}
 \begin{split}
 \E[\xi_k|\hat\mO_t] &= \begin{cases}
 \E[\xi_k|\hat{\xi}_k,\theta_k], ~~~~~~&\text{if   } \hat{\xi}_k \in \hat{\mO}_t,\\
 E[\xi_k],& \text{otherwise}.
 \end{cases}  \\
 &= \vartheta_{k,t}\bar{\xi}_k,
 \end{split}
 \end{align}
 where we have used the definitions of $\bar\xi_t$ and $\vartheta_{k,t}$ to compactly write $\E[\xi_k|\hat\mO_t]=\vartheta_{k,t}\E[\xi_k|\hat{\xi}_k,\theta_k]=\vartheta_{k,t}\bar{\xi}_k$. 
 From this observation, and using Lemma \ref{L:inno}, the expression of $\bar{X}_t$ is computed in the following lemma.

\begin{lm} \label{L:barx}
For any $t$, $\bar{X}_t=\E[X_t|\I^c_t]$ is given by,
\begin{align}\label{E:xtilde}
\bar{X}_t=A^t\mu_0+\sum_{k=0}^t\Psi(t,k)\vartheta_{k,t}\bar\xi_k+\sum_{k=0}^{t-1}A^{t-1-k}BU_k.
\end{align}
\end{lm}  
 
\begin{proof}
Notice that, from \eqref{E:x_bar_new} we have
\begin{align*}
\E[X_t|\I^c_t]=&A^t\mu_0+\sum_{k=0}^t\Psi(t,k)\vartheta_{k,t}\E[\xi_k|\hat{\mO}_t]+\sum_{k=0}^{t-1}A^{t-1-k}BU_k.
\end{align*}
The lemma follows immediately after we substitute the expression of $\E[\xi_k|\hat{\mO}_t]$ from \eqref{E:xi_expected} into the last equation.
\end{proof} 
 \vspace*{6 pt}
 
Define the error $e_t\triangleq X_t-\bar{X}_t $.
It follows from  \eqref{E:xtilde} that
\begin{align*}
    e_t = A^t X_0+\sum_{k=0}^{t-1} A^{t-k-1}W_{k} - A^t\mu_0 - \sum_{k=0}^t\Psi(t,k)\vartheta_{k,t}\bar\xi_k.
\end{align*}
Notice that $e_t$ does not depend on the control strategy $\gamma^\U$. 
However, it does depend on the quantizer selection strategy $\gamma^\Theta$ through the last term in the above equation.
Furthermore, for all $t$, $\E[e_t] =0 $  since $\E[\bar X_t] =\E[X_t] $ due to the law of total expectation. 

 
 

 At this point we are ready to return to the cost function \eqref{E:cost2} and find the optimal controller and the optimal quantizer selection policies.

Associated with the cost function \eqref{E:cost2}, let us define the value function as follows:
\begin{subequations}\label{E:preV}
\begin{align}
V_k(\I_k) = &\min_{\{\gamma^u_t\}_{t=k}^{T-1},\{\gamma^\theta_t\}_{t=k}^{T-1}}\E_\gamma\Big[\sum_{t=k}^{T-1}(X_t\T Q_1X_t+U_t\T RU_t+\theta_t\T \Lambda) +X_T\T Q_2X_T\Big],\\
V_T(\I_t)=&\E_\gamma[X_T\T Q_2X_T],
\end{align}
\end{subequations}
where the information set $\I_k = \{\I^c_k, \bar\I^q_k\}$ and $\E_\gamma[\cdot]$ denotes the expectation under the strategy pair $\gamma =(\gamma^\U,\gamma^\Theta)$.
Using the dynamic programming principle,
\begin{align} \label{E:V}
V_k(\I_k)=&\min_{\gamma^u_k\in \Gamma^u_k,\gamma^\theta_k\in \Gamma^\theta_k}\E_\gamma\Big[(X_k\T Q_1X_k+U_k\T RU_k+\theta_k\T \Lambda) +V_{k+1} \Big].
\end{align}

If $\gamma^{u*}_k$ and $\gamma^{\theta*}_k$ minimize the right-hand-side of \eqref{E:V}, then the optimal strategies are $U_k^*=\gamma^{u*}_k(\I^c_k)$ and $\theta_k^*=\gamma^{\theta *}_k(\bar\I^q_k)$.
From \eqref{E:preV}, we also have that
\begin{align} \label{E:exVal}
\min_{\gamma^\U\in \Gamma^\U,\gamma^\Theta\in \Gamma^\Theta}J(\gamma^\U,\gamma^\Theta)=\E_\gamma[{V_0}].
\end{align}
The following  theorem characterizes the optimal policy $\gamma^{u*}_k(\cdot)$ for all $k=0,1,\ldots,T-1$.
\begin{thm}[Optimal Control Policy] \label{T:optcont}
Given the information $\I^c_k$ to the controller at time $k$, the optimal control policy $\gamma^{u*}_k:\I^c_k\to\R^m$ that minimizes the right-hand-side of \eqref{E:V} has the following structure
\begin{align}
U^*_k=\gamma^{u*}_k(\I^c_k)=-L_k\bar{X}_k,
\end{align}
where $\bar{X}_k$ is computed in Lemma~\ref{L:barx}  for all $k =0,1,\ldots,T-1$, and the matrices $L_k$ and $P_k$ are obtained by
\begin{subequations}
\begin{align}
\begin{split}
L_k=(R+B\T P_{k+1}B)^{-1}B\T P_{k+1}A, \label{E:eqLk}
\end{split}\\
\begin{split}
P_k= Q_1+A\T P_{k+1}A-L_k\T (R+B\T P_{k+1}B)L_k, \label{E:eqPk}
\end{split}\\
\begin{split}
P_T=Q_2.
\end{split}
\end{align}
\end{subequations}
\end{thm}
\vspace{10 pt}

\begin{proof}
The proof of this theorem is based on the dynamic programming principle.
Specifically, 
if there exist value functions $V_k$ for all $k=0,1,\ldots,T$ that satisfy \eqref{E:V}, then the optimal control $U_k^*$ and the optimal quantizer selection $\theta_k^*$ are obtained by the policies $\gamma^{u*}_k$ and $\gamma^{\theta*}_k$ that minimize \eqref{E:V}. 

Let us assume that the value function at time $k =0,1,\ldots,T-1$ is of the form:
\begin{align} \label{E:hypoV}
V_k(\I_k)=\E_\gamma[X_k\T P_k X_k]+C_k+r_k,
\end{align}
where $P_k$ is as  in \eqref{E:eqPk}, and, for all $k=0,1,\ldots,T-1$,
\begin{align} \label{E:ck_old}
C_k=\min_{\{\gamma^\theta_t\}_{t=k}^{T-1}}\E_{\gamma^\theta}\left[\sum_{t=k}^{T-1}e_t\T N_t e_t+\theta_t\T \Lambda\right],
\end{align}
 where $N_k \in \R^{n\times n}$ and $r_k\in \R$ are given by
 \begin{subequations}
\begin{align}
N_k=&L_k\T (R+B\T P_{k+1}B)L_k, \label{E:eqNk} \\ 
r_k=&r_{k+1}+\tr(P_{k+1}\W), \label{E:eqrk} \\
r_T=&0.
\end{align}
\end{subequations}

Equation \eqref{E:ck_old} can be re-written as
\begin{align*}
C_{k}=\min_{\gamma^\theta_k}\E_{\gamma^{\theta}}\left[e_k\T N_k e_k+\theta_k\T \Lambda+C_{k+1} \right].
\end{align*}

We first verify that $V_{T-1}$ is of the form \eqref{E:hypoV}.
\begin{align} \label{E:eqVT-1}
V_{T-1}=\min_{\gamma^u_{T-1},\gamma^\theta_{T-1}}\E_\gamma\Big[&X_{T-1}\T Q_1X_{T-1} +U_{T-1}\T RU_{T-1}+\theta_{T-1}\T \Lambda+X_T\T P_TX_T \Big].
\end{align}
Substituting the equation $X_T=AX_{T-1}+BU_{T-1}+W_{T-1}$, and after some simplifications, yields
\begin{align*}
V_{T-1}=\min_{\gamma^u_{T-1},\gamma^\theta_{T-1}}\E_\gamma \Big[\|U_{T-1}&+L_{T-1}X_{T-1}\|^2_{(R+B\T P_TB)}+X_{T-1}\T P_{T-1}X_{T-1} \\
&+\theta_{T-1}\T \Lambda+\tr(P_T\W) \Big],
\end{align*}
where $\|L\|^2_K \triangleq L\T K L$ for any two matrices $L$ and $K$ of compatible dimensions.
In the previous expression, $\|U_{T-1}+L_{T-1}X_{T-1}\|^2_{(R+B\T P_TB)}$ is the only term that depends on $U_{T-1}$. 
Therefore, we seek $\gamma^u_{T-1}:\I^c_{T-1}\to \R^m$  that minimizes the mean-square error $\E\left[\|U_{T-1}+L_{T-1}X_{T-1}\|^2_{(R+B\T P_TB)}\right]$. 
Thus, the optimal $U_{T-1}$ is a minimum mean squared estimate of $-L_{T-1}X_{T-1}$ based on the $\sigma$-field generated by $\I^c_{T-1}$. 
Hence, from Lemma~\ref{L:ortho},
\begin{align}
U_{T-1}^*&=\gamma^{u*}_{T-1}(\I^c_{T-1})=-L_{T-1}\E[X_{T-1}|\I^c_{T-1}] 
=-L_{T-1}\bar X_{T-1}.
\end{align} 
After substituting the optimal $U^*_{T-1}$ in \eqref{E:eqVT-1}, we obtain
\begin{align*}
V_{T-1}=\min_{\gamma^\theta_{T-1}}\E_\gamma\Big[&\|X_{T-1}-\bar X_{T-1}\|^2_{N_{T-1}}+\theta_{T-1}\T \Lambda+\tr(P_T\W)+X_{T-1}\T P_{T-1}X_{T-1}\Big].
\end{align*}
The above expression of $V_{T-1}$ can be rewritten as follows
\begin{align*}
V_{T-1}=&\min_{\gamma^\theta_{T-1}}\E_{\gamma^\theta}\left[e_{T-1}\T N_{T-1}e_{T-1}+\theta_{T-1}\T \Lambda\right]+\E[X_{T-1}\T P_{T-1}X_{T-1}]+\tr(P_T\W).
\end{align*}
Therefore, using the definitions of $C_{T-1}$ and $r_{T-1}$ from \eqref{E:ck_old} and \eqref{E:eqrk}, we obtain
\begin{align*}
V_{T-1}=\E[X_{T-1}\T P_{T-1}X_{T-1}]+C_{T-1}+r_{T-1}.
\end{align*}
Thus, $V_{T-1}$ is of the form \eqref{E:hypoV}. 
Next, we prove the hypothesis \eqref{E:hypoV} using mathematical induction.
To that end, we now assume that \eqref{E:hypoV} is true for some $k+1$. 
Then,
\begin{align*}
V_k=\min_{\gamma^u_k,\gamma^\theta_k}\E_\gamma\Big[&(X_k\T Q_1X_k+U_k\T RU_k+\theta_k\T \Lambda) \nonumber+V_{k+1}~\big]\\
=\min_{\gamma^u_k,\gamma^\theta_k}\E_\gamma\Big[&(X_k\T Q_1X_k+U_k\T RU_k+\theta_k\T \Lambda) \nonumber+X_{k+1}\T P_{k+1}X_{k+1}+r_{k+1}+C_{k+1}~\big].
\end{align*} 

Using \eqref{E:dyn}, and after some simplifications, it follows that
\begin{align} \label{E:eqVk}
V_k=\min_{\gamma^u_{k},\gamma^\theta_{k}}\E_\gamma\Big[\|U_{k}&+L_{k}X_{k}\|^2_{(R+B\T P_{k+1}B)}+X_{k}\T P_{k}X_{k} +\theta_{k}\T \Lambda\\ \nonumber
& +\tr(P_{k+1}\W)+r_{k+1}+C_{k+1} \Big].
\end{align}
One may notice from the definition of $e_k$ that it does not depend on the past control history $\U_k$, but rather, it depends on the quantizer selection history $\Theta_k$. 
Thus, $C_k$ does not depend on the control history $\U_k$.
Furthermore, from \eqref{E:eqNk}, \eqref{E:eqrk} and \eqref{E:eqPk}, one notices that  $N_k$, $r_k$ and $P_k$ 
do not depend on the past (or future) decisions on the control or quantizer-selection.
Therefore, $\|U_{k}+L_{k}X_{k}\|^2_{(R+B\T P_{k+1}B)}$ is the only term in the above expression of $V_k$ that depends on $U_k$.
Using Lemma \ref{L:ortho}, the optimal $\I_k^c$-measurable control $U_k^*$ that minimizes  $\E\left[ \|U_{k}+L_{k}X_{k}\|^2_{(R+B\T P_{k+1}B)}\right]$ is given by
\begin{align} \label{E:u*}
U_k^*=\gamma^{u*}_k(\I_k^c)=-L_k\E\left[X_k|\I^c_k\right]=-L_k\bar X_k.
\end{align} 
After substituting the optimal control in \eqref{E:eqVk}, we obtain
\begin{align*}
V_k=&\E[X_k\T P_k X_k]+\min_{\gamma^\theta_k}\E_\gamma\Big[e_k\T (L_k\T (R+B\T P_{k+1}B)L_k)e_k+\theta_k\T \Lambda+C_{k+1}\Big]\\
&+\tr(P_{k+1}\W)+r_{k+1}\\
=&\E[X_k\T P_k X_k]+\min_{\gamma^\theta_k}\E_{\gamma^\theta}\Big[e_k\T N_k e_k+\theta_k\T \Lambda+C_{k+1}\Big]+r_k\\
=&\E[X_k\T P_k X_k]+C_k+r_k.
\end{align*}
Thus, the value function is indeed of the form \eqref{E:hypoV}, and hence, the optimal control at time $k=0,1,\cdots,T-1$ is given by \eqref{E:u*}.
This completes the proof.
\end{proof}

\begin{rem}
From Theorem \ref{T:optcont}, the optimal control is linear in  $\bar{X}_k$. 
The optimal gain  $-L_k$ can be computed offline without knowledge of $\gamma^{\Theta*}$. 
The effect of $\gamma^{\Theta*}$ on $\gamma^{\U*}$ is through the term $\bar{X}_k$, which can be computed online using \eqref{E:xtilde}.
\end{rem}

Having computed the optimal controller, we now focus on solving for the optimal selection of the quantizers. 
To that end, from \eqref{E:hypoV}, we have
\begin{align*}
V_0=\E[X_0\T P_0X_0]+C_0+r_0,
\end{align*}
and thus,
\begin{align*}
\min_{\gamma^\U\in \Gamma^\U,\gamma^\Theta\in \Gamma^\Theta}J(\gamma^\U,\gamma^\Theta)= \E[{V_0}]=\mu_0\T P_0\mu_0 +\tr(P_0\Sigma_x)+r_0+C_0,
\end{align*}
where,  from \eqref{E:ck_old}, $C_0$ can be written as
\begin{align} \label{E:ck}
C_0=&\min_{\{\gamma^\theta_t\}_{t=0}^{T-1}}\E_{\gamma^\theta}\left[\sum_{t=0}^{T-1}e_t\T N_te_t+\theta_t\T \Lambda\right].
\end{align}
Notice that the effect of the quantizer-selection policy $\gamma^\Theta$ on the cost $J(\gamma^\U,\gamma^\Theta)$ is reflected only through the term $C_0$. 
The optimal quantizer selection policy can thus be found by performing the minimization associated with $C_0$ as represented in \eqref{E:ck}.

\subsection{Optimal Quantizer Selection Policy} \label{sec:optimal_partial_infor_quantizer_selection}

In this section, we study the optimal quantizer-selection policy $\gamma^{\Theta*}$, which can be found by solving \eqref{E:ck}.
We may write $\E[e_t\T N_t e_t]=  \tr(N_t\E[e_t e_t\T])$, and the following Lemma computes $\E[e_te_t\T]$.

\begin{lm} \label{L:et_second_moment}
For all $t\in \bbN_0$,
\begin{align*}
    \E[e_te_t\T] =  \Sigma_t\! +\!\! \sum_{k=0}^t \!\Psi(t,k) (M_k -\E[\vartheta_{k,t}F_k(\theta_k)])\Psi(t,k)\T.
\end{align*}
\end{lm}
 \begin{proof}
 The proof is given in Appendix~\ref{A:delta}.
 \end{proof}
 \vspace*{6 pt}
 
Using Lemma~\ref{L:et_second_moment}, 
%
%
%
%
Therefore, the cost $C_0$ can be simplified as
\begin{align} \label{E:MINP}
&C_0=\sum_{t=0}^{T-1}\left(\tr(\Sigma_tN_t)+\sum_{k=0}^t\tr(\tilde N_{k,t}M_k)\right)+\min_{\{\gamma^\theta_t\}_{t=0}^{T-1}}\E_{\gamma^\theta}\left[\sum_{t=0}^{T-1}\tr\left(\Pi_t(\Theta)F_t(\theta_t)\right)+\theta_t\T\lambda\right],
\end{align}
where
\begin{subequations}
\begin{align}
    \tilde N_{k,t}=\Psi(t,k)\T N_t \Psi(t,k), \label{E:ntilde} \\
\Pi_t(\Theta)=-\sum_{\ell=t}^{T-1}\vartheta_{t,\ell}\tilde{N}_{t,\ell}. \label{E:piup}
\end{align}
\end{subequations}
The optimal quantizer selection policy is found by solving the  Mixed-Integer-Nonlinear Program (MINP)  in \eqref{E:MINP}.
 
At this point it may appear that the expression $\sum_{t=0}^{T-1}\tr(\Pi_t(\Theta)F_t(\theta_t))$ in \eqref{E:MINP} is a nonlinear function of $\Theta$. 
However, we now show that after some simplifications, it can be written as a linear function of $\Theta$. 
By expressing \eqref{E:MINP}    as a linear function of $\Theta$, we can recast \eqref{E:MINP} as a Mixed-Integer-Linear-Program (MILP), which further can be solved efficiently using existing efficient solvers \cite{lofberg2004yalmip}.

To express \eqref{E:MINP} as an MILP, we construct a matrix $\Phi\in \R^{T\times M}$ as follows: for all $i=0,\ldots,T-1$ and $j=1,\ldots,M$, let
\begin{align} \label{E:phi}
[\Phi]_{ij}=\begin{cases}
1, & ~~{\textrm{  if }} i \ge d_j,\\
0, &~~\textrm{otherwise},
\end{cases}
\end{align}
where $[\Phi]_{ij}$ is the $ij$-th component of $\Phi$ matrix. It directly follows from the definition of $\Phi$ that $1_{d_j\le t-k}=[\Phi]_{t-k,j}$. 
Consequently, we can express  \eqref{E:vartheta} as
\begin{align*}
\vartheta_{k,t}=\sum_{i=1}^M\theta_k^i[\Phi]_{t-k,i}.
\end{align*}

Thus, $\Pi_t(\Theta)$ in \eqref{E:piup} can be rewritten  as $\Pi_t(\Theta)=-\sum_{\ell=t}^{T-1}\sum_{i=1}^M\theta_t^i[\Phi]_{\ell-t,i}\tilde{N}_{t,\ell}$. 
Also, from \eqref{E:F}, we have that $F_t(\theta_t)=\sum_{i=1}^M\theta_t^iF^i_t$. Thus,
\begin{align*}
\tr(&\Pi_t(\Theta)F_t(\theta_t))=-\tr\left(\sum_{i=1}^M(\theta_t^i\sum_{\ell=t}^{T-1}[\Phi]_{\ell-t,i}\tilde{N}_{t,\ell})F_t(\theta_t)\right)\\
=&-\tr\left(\left(\sum_{i=1}^M\bigg(\theta_t^i\sum_{\ell=t}^{T-1}[\Phi]_{\ell-t,i}\tilde{N}_{t,\ell}\bigg)\right)\left(\sum_{j=1}^M\theta_t^jF^j_t\right)\right)\\
\stackrel{(a)}{=}&-\tr\left(\sum_{i=1}^M\theta_t^i\bigg(\sum_{\ell=t}^{T-1}[\Phi]_{\ell-t,i}\tilde{N}_{t,\ell}\bigg)F^i_t\right)=-\sum_{i=1}^M\theta_t^i\beta_t^i,
\end{align*}
where $\beta_t^i=\tr\left(\left(\sum_{\ell=t}^{T-1}[\Phi]_{\ell-t,i}\tilde{N}_{t,\ell}\right)F^i_t\right)$ and $(a)$ follows from the fact $\theta_t^i\theta_t^j=0$ if $i\ne j$.
Note that the coefficients $\beta_t^i$ can be computed offline.

From the previous derivation, $C_0$ in \eqref{E:MINP} becomes 
\begin{align} \label{E:C_MILP}
&C_0=c+\min_{\{\gamma^\theta_t\}_{t=0}^{T-1}}\E_{\gamma^\theta}\left[\sum_{t=0}^{T-1} c_t\T \theta_t\right],
\end{align}
where constant $c=\sum_{t=0}^{T-1}\left(\tr(\Sigma_tN_t)+\sum_{k=0}^t\tr(\tilde N_{k,t}M_k)\right)$ and $c_t=[c_t^1,\ldots,c_t^M]\T$ with $c_t^i= \lambda_i-\beta^i_t$.
Notice that, in \eqref{E:C_MILP}, the cost function is linear in $\theta$ and the coefficients $c^i_t$ are deterministic (and can be computed offline). 
Therefore, it is sufficient to look for a deterministic strategy to minimize the linear cost $\sum_{t=0}^{T-1} c_t\T \theta_t$, as the class of deterministic strategies contains an optimal solution for $\min_{\{\gamma^\theta_t\}_{t=0}^{T-1}}\E_{\gamma^\theta}\left[\sum_{t=0}^{T-1} c_t\T \theta_t\right]$. 
The following lemma presents an MILP formulation to obtain the optimal quantizer selection policy.
\medskip

\begin{lm}
The optimal quantizer selection strategy is found by solving the following Mixed-Integer-Linear-Program
\begin{subequations}\label{E:MILP}
\begin{align}
\min_{\Theta} &\sum_{t=0}^{T-1}c_t\T\theta_t,\\
\rm{s.t.  } ~~~~~~&\theta_t^i \in \{0,1\}, ~~~~~~ t=0,\ldots, T-1, ~~~i=1,\ldots,M,\\
&\sum_{i=1}^M\theta_t^i=1, ~~~~~~ t=0,\ldots, T-1.
\end{align}
\end{subequations}
\end{lm}

\begin{proof}
The proof follows directly from the derivation of \eqref{E:C_MILP} and the subsequent discussion.
\end{proof}
Notice that in \eqref{E:MILP} there is no constraint coupling $\theta_k$ and $\theta_\ell$, and the cost function in \eqref{E:MILP} is also decoupled in $\theta_k$ and $\theta_\ell$ for all $k\ne \ell \in \{0,\ldots,T-1\}$.
 Therefore, the optimal $\theta_t$ at time $t$ can be found by minimizing $c_t\T\theta_t$ subject to the constraints $\sum_{i=1}^M\theta_t^i=1$, $\theta^i_t \in \{0,1\}$.
  Thus, the optimal quantizer selection strategy for this problem turns out to be remarkably simple: if $i^*=\stackrel[i=1,\ldots,M]{}{\arg\min}\{c^1_t,\ldots,c^M_t\}$, then the optimal strategy is to use the $i^*$-th quantizer\footnote{In case there exists multiple minimizers for $\stackrel[i=1,\ldots,M]{}{\arg\min}\{c^1_t,\ldots,c^M_t\}$, one of these minimizers can be chosen randomly without affecting the optimality.} such that
\begin{align*}
\gamma^{\theta*}_t=\theta_t^*=[1_{i^*=1},\ldots,1_{i^*=M}]\T.
\end{align*}
This result is summarized in the following theorem.
\medskip

\begin{thm}[Optimal Quantizer Selection]\label{T:optquant}
At time $t$, the $j$-th quantizer is optimal  if and only if
\begin{align*}
c^{j}_t = \min\{c^1_t,\ldots,c^M_t\},
\end{align*}
where, for all $i=1,\ldots,M$,
\begin{align*}
c^i_t=\lambda_i-\tr\left(\left(\sum_{\ell=t}^{T-1}[\Phi]_{\ell-t,i}\tilde{N}_{t,\ell}\right)F^i_t\right).
\end{align*}
and $\tilde{N}_{t,\ell},[\Phi]_{\ell-t,i}$ and $F^i_t$ are defined in equation \eqref{E:ntilde}, \eqref{E:phi} and \eqref{E:Fi} respectively.
\end{thm}
The following remark is immediate from Theorem \ref{T:optquant}.

\medskip

\begin{rem}
  The optimal strategy for selecting the quantizers can be computed offline. This requires an offline computation of $\tilde{N}_{t,\ell}$ and $F^i_t$, but it does not require knowledge of the optimal control strategy. 
  \end{rem}


\subsection{Discussion and Remarks}

We delve into the cost $c_t\T \theta_t$ in \eqref{E:MILP} to discuss how the three factors,  namely, the cost of quantization, the quantization resolution, and the delay, affect the cost function.
The coefficients $c^i_t$ which determine the optimal quantizer selection strategy at time $t$ have two components, namely, $\lambda_i$, and $\beta^i_t$, where $\lambda_i$ is the cost for using the $i$-th quantizer, and $\beta^i_t$ captures the trade-off between quantization quality and the associated delays. 
Let us discuss each of these two terms in greater detail. 
First, $c^i_t$ being proportional to the cost $\lambda_i$,  reflects the fact that lower quantization cost is desirable.
 The quantity $\beta^i_t$ is arguably more interesting. 
 Note that $\beta^i_t$ is of the form $\tr(G^i_tF^i_t)$, where for all $i$, $G^i_t$ is a positive (semi)-definite matrix whose expression can be easily identified from the expression of $\beta^i_t$. 
 Moreover, since  $1\ge [\Phi]_{i,1}\ge [\Phi]_{i,2} \ge \ldots \ge [\Phi]_{i,M} \ge 0$ for all $i=0,\ldots,T-1$, we have $G^1_t\succeq G^2_t\succeq \ldots \succeq G^M_t$.
  On the other hand, by using the $i$-th quantizer, the reduction in \textit{uncertainty covariance} is $F^i_t$.
   By \textit{uncertainty covariance} we mean the following: Before the arrival of any measurement ($\hat \xi_t$), $\xi_t$ is a Gaussian distributed random variable with covariance $M_t$.
   Once a quantized version ($\hat{\xi}_t$) of $\xi_t$ arrives at the controller, the controller receives information on the realization of the random variable $\xi_t$. 
   Specifically, at this point, the controller knows in which of the $\p^i_j \subset \R^p$ the random variable $\xi_t$ belongs to. 
   Therefore, the posterior distribution of $\xi_t$ changes after receiving $\hat{\xi}_t$, and the difference between the covariance of this posterior distribution and the prior distribution is $F^i_t$ if the $i$-th quantizer is used. 
   Needless to say,  had there been a quantizer which could ensure $\hat{\xi}_t=\xi_t$, i.e., no loss during quantization for every realization of $\xi_t$, then the reduction in covariance is exactly $M_t$ and the posterior distribution of $\xi_t$ at the controller is a Dirac measure around $\hat{\xi}_t$. 
   Use of quantized measurements is similar as operating somewhere in between open-loop and closed-loop control. 
   In open-loop, no measurement is sent, and in closed-loop, the exact measurement is sent without any distortion. 
   By means of  quantization, the controller receives \textit{something} but not \textit{everything}. 
Furthermore,   since $\beta^i_t \ge 0$ and since it appears with a negative sign in the cost function,  it is clearly desirable to choose a quantizer that would maximize $\beta^i_t$. 
  The matrix $F^i_t$ directly reflects how much reduction in covariance will occur if the $i$-th quantizer is used. 
  The matrix $G^i_t$, on the other hand, incorporates the delay associated with the $i$-th quantizer. 
  As $i$ is increased from $1$ to $M$, $G^i_t$ decreases $\tr(G^i_tF^i_t)$, reflecting the fact that smaller delay is preferable.
  However, as $i$ is varied, $F^i_t$ shows the variation in covariance reduction. For example, if  reduction  in covariance increases with the increase in $\ell_i$, then $F^i_t$ is attempting to increase $\tr(G^i_tF^i_t)$ as $i$ is varied from $1$ to $M$.
  Thus, there is a dual behavior between $F^i_t$ and $G^i_t$ as $i$ changes, and this duality is captured by the parameters of the channel and the quantizers, namely, $\p^i$, $\ell^i$, and the delay $d_i$.

  We conclude this section with a few more remarks.

  \begin{rem}
  The cost function in \eqref{E:MILP} resembles the component $\sum_{t=0}^{T-1} \Lambda\T \theta_t$ in \eqref{E:cost2}, except that all the state and control costs are absorbed in the coefficients $c^i_t$. Here $c^i_t$ can be viewed as the \textit{adjusted cost} for operating the  $i$-th quantizer at time $t$, and the adjustment factor is $\beta_t^i$ which can be computed offline.
  \end{rem}
  
\begin{rem} \label{R:openloop}
This approach allows for the case when the set of available quantizers contains a quantizer $\Q^0$ with only one quantization level, i.e., $\ell_0=1$, $\p^0=\{\p^0_{1}=\R^p\}$, and quantization cost $\lambda_0=0$. This quantizer produces the same quantized output for every input signal, hence providing the option to  remain open-loop. For such a quantizer, it can be verified from \eqref{E:Fi} that $F^0_t=0$ for all $t$. Therefore, $c^0_t=\lambda_0-\beta^0_t=0$ for all $t$, and selection of this quantizer at any time $t$ reflects the fact that it is optimal not to send any information to the controller at that time.
If the quantization costs are very high \footnote{Or, the quantization cost is higher than the reward from using quantization, i.e. $\lambda_i > \beta^i_t$ for all $t$.} $\lambda_i \gg 1$,  the optimal choice of the quantizers would be $\Q^0$, and hence, the controller will not be receiving any information, which in principle, is equivalent to open-loop control.
\end{rem}

  \section{Special Cases} \label{S:special}
In this section we consider two special cases, namely: (i) constant-delay case,  and (ii) full observation case.  
  
  \subsection{Constant-Delay}
  
  In this section we consider the case where $d_1=d_2=\ldots=d_M=d$, i.e., the delay induced by each quantizer is the same.
   Intuitively, since the delay is not affected by the choice of the quantizer, then the quantizer selection problem should reduce to a trade-off between the quantization cost and the quality of quantization. To see this, let us first note that $[\Phi]_{i,1}=\ldots=[\Phi]_{i,M}=1_{i\ge d}$ for all $i=0,\ldots,T-1$. Therefore,
  \begin{align*}
  \beta^i_t=&\tr\left(\left(\sum_{\ell=t}^{T-1}[\Phi]_{\ell-t,i}\tilde{N}_{t,\ell}\right)F^i_t\right)\\
  =&\tr\left(\left(\sum_{\ell=t}^{T-1}1_{\ell-t\ge d}\tilde{N}_{t,\ell}\right)F^i_t\right)
  =\tr\left(\left(\sum_{\ell=t+d}^{T-1}\tilde{N}_{t,\ell}\right)F^i_t\right)\\
  =&\tr\left(H(t,d)F^i_t\right),
  \end{align*}
  where $H(t,d)=\sum_{\ell=t+d}^{T-1}\tilde{N}_{t,\ell} \succeq 0$. 
  Thus, for fixed $t$ and $d$, whether the $i$-th quantizer is optimal at time $t$ is entirely determined by $F^i_t$ where recall that $F^i_t$ represents the \textit{uncertainty covariance} reductions.
  
  Also notice that $H(t,d)=0$ for all $t\ge T-d $, and hence $\beta^i_t=0$. 
  Therefore, the optimal selection for the quantizers for $t\ge T-d$ would be the one with the lowest $\lambda_i$.
   This is due to the fact that the quantized information $\xi_{T-d},\xi_{T-d+1},\ldots$ will not be available at the controller before time $T-1$, and hence these quantized measurements would be of no use to the controller. 
   Therefore, the quality of the quantization for time $T-d$ onward is immaterial to the controller, and hence, the lowest cost quantizer would be the optimal.
  
  \subsection{Full Observation}
  
  For the full observation case we substitute $\V=0$ and $C=I$ in the analysis presented above. As a direct consequence, one can verify that $\xi_t=W_{t-1}$  for all $t$.
Therefore, $\{\xi_t\sim \N(0,\W)\}_{t\in\bbN_0}$ are i.i.d signals, and consequently  the matrices $F^i_t$ given in \eqref{E:F} will be time invariant, i.e., $F^i_1=\ldots=F^i_{T}\triangleq F^i$.

  For all $t\in \bbN_0$, $\Sigma_t=0$, $\Sigma_{t+1|t}=M_{t+1}=\W$. This also implies that, for all $t \ge k$,
  \begin{align*}
  \Psi(t,k)=A^{t-k}, \text{ and }
  \tilde{N}_{k,t}={A^{t-k}}\T N_tA^{t-k}.
  \end{align*}
  Therefore, the state estimate can be written as
  \begin{align} \label{E:estimate_cosnt}
  \tilde{X}_t=&A^t\mu_0+\sum_{k=0}^t\Psi(t,k)\vartheta_{k,t}\bar{\xi}_k+\sum_{k=0}^{t-1}A^{t-1-k}BU_k\nonumber \\
  =&A^t\mu_0+\sum_{k=0}^tA^{t-k}\vartheta_{k,t}\bar{\xi}_k+\sum_{k=0}^{t-1}A^{t-1-k}BU_k\nonumber\\
  =&A\tilde{X}_{t-1}+BU_{t-1}+\vartheta_{t,t}\bar{\xi}_t+\sum_{k=0}^{t-1}A^{t-k}(\vartheta_{k,t}-\vartheta_{k,t-1})\bar{\xi}_k.
  \end{align}
   
   The expression for $\beta^i_t$ is now given by:
   \begin{align*}
   \beta^i_t=&\tr\left(\left(\sum_{\ell=t}^{T-1}[\Phi]_{\ell-t,i}\tilde{N}_{t,\ell}\right)F^i_t\right)\\
   =&\tr\left(\left(\sum_{\ell=t+d_i}^{T-1}{A^{\ell-t}}\T N_\ell A^{\ell-t}\right)F^i\right).
   \end{align*}
   
Let us define a symmetric matrix $\Upsilon_t$ as follows
\begin{align*}
\Upsilon_{t}=&A\T \Upsilon_{t+1}A+N_t,\\
\Upsilon_T=&0,
\end{align*}   
which allows us to rewrite $\beta^i_t=\tr(\Upsilon_{\min\{t+d_i,T\}}F^i)$.
We conclude this section by discussing the constant delay case for fully observed systems.

Under the assumption of constant delay, i.e., $d_1=\ldots=d_M=d$, we obtain $\beta^i_t=\tr(\Upsilon_{\min\{t+d,T\}}F^i)$.
Furthermore,  $\vartheta_{k,t}=1$ if and only if $t-k\ge d$, otherwise $\vartheta_{k,t}=0$. This implies from \eqref{E:estimate_cosnt} that, for all $t\in \mathbb N_0$,
\begin{align}
\bar{X}_t=\begin{cases} A\bar{X}_{t-1}+BU_{t-1}+A^d\bar{\xi}_{t-d}, ~~~& \textrm{if  } t\ge d,\\
A\bar{X}_{t-1}+BU_{t-1}, & \textrm{otherwise}.
\end{cases}
\end{align}

\section{Numerical Examples}
In this section, we illustrate our theory on the following system
\begin{subequations}
\begin{align}
    X_{t+1}&=\begin{bmatrix}
1.01 &0.5\\ 0 &1.1
\end{bmatrix}X_t+\begin{bmatrix}
0.1 &0\\ 0 &0.15
\end{bmatrix}U_t+W_t,\\
Y_{t}&=\begin{bmatrix}
1 &0\\ 1 &1
\end{bmatrix}X_t+\nu_t,
\end{align}
\end{subequations}
where $X_0\sim \N(0,I)$, $W_t \sim \N(0,\frac{1}{2}I)$, and $\nu_t \sim \N(0,\frac{1}{4}I)$.
The control cost has  parameters $Q=Q_f=R=\frac{1}{2}I$, and the time horizon was set to $T=50$.
\begin{figure}
    \centering
    \includegraphics[draft=false, width=0.7\linewidth]{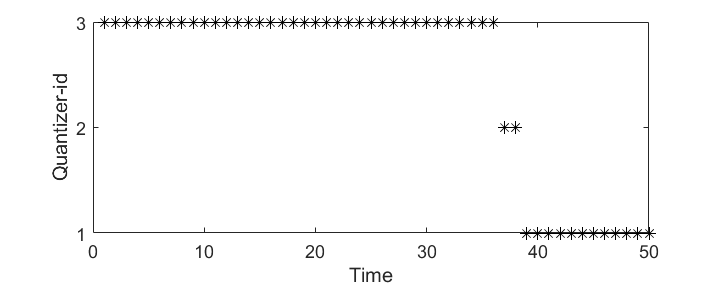}
    \caption{Optimal selection of the quantizers over time when $d_i =i$ for all $i$.}
    \label{F:rb1}
\end{figure}
\begin{figure}
    \centering
    \includegraphics[draft=false, width=0.7\textwidth]{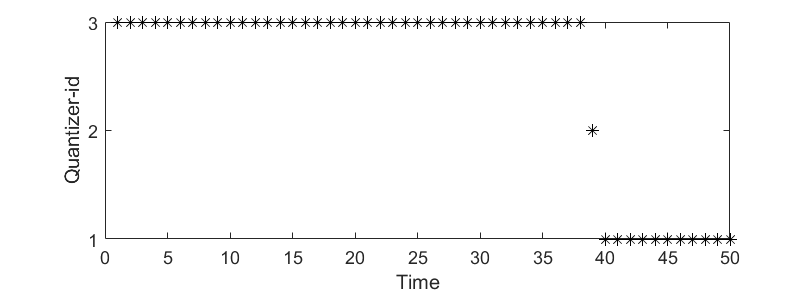}
    \caption{Optimal selection of the quantizers over time when $d_i =1$ for all $i$.}
    \label{F:rb3}
\end{figure}

The simulation was performed with three quantizers ($\Q^1,\Q^2,\Q^3$) where $\Q^i$ has $2^i$ number of quantization levels. 
The partitions associated with the quantizers are $\p^1=\{\R_+\times\R, \R_{< 0}\times \R\}$, $\p^2=\{\R_+\times \R_+,~\R_+\times\R_{< 0},~ \R_{<0}\times\R_+, \R_{<0}\times\R_{<0}\}$ and $\p^3=\{[0,1)\times\R_+, ~[1,\infty)\times\R_+,[0,1)\times\R_{<0}, ~[1,\infty)\times\R_{<0},[-1,0)\times\R_+, ~(-\infty,-1)\times\R_+,[-1,0)\times\R_{<0}, ~(-\infty,-1)\times\R_{<0}\}$. The costs associated with the quantizers are $\Lambda=[100 , 200, 300]\T$. 

We consider two scenarios where in the first scenario the delays associated with the quantizers are $d_i=i$ for all $i$, and in the second scenario, $d_i =i$ for all $i$.
Under these conditions the optimal selections for the quantizers are plotted in Figures \ref{F:rb1} and \ref{F:rb3} respectively. 
Although Figures \ref{F:rb1} and \ref{F:rb3} portray similar behavior, there are minor differences in the optimal selection of the quantizers due to the delays.  
 For example, from Figures \ref{F:rb1} and \ref{F:rb3}, one notices that at $t=37$, $\Q^3$ is optimal when $d_3=1$, whereas $\Q^2$ is optimal when $d_3=3$. 
The reason behind this is the fact that the quantized output of both $\Q^3$ and $\Q^2$ will be available with same delay when $d_i = 1$ for all $i$, whereas the quantized output of $\Q^3$ will reach later than that of $\Q^2$ when $d_i= i$, although the quantized output of $\Q^3$ will less distorted than that of $\Q^2$.
At this particular instance, it turned out to have coarser measurement faster than finer measurement with more delayed.
{Thus, this simple example reflects the combined (dual) effect of the quantization resolution and the associated delays in the optimal choice of the quantizers.}

The same example is considered when $\nu_t=0$ for all $t$, i.e., a perfect state feedback scenario. 
The optimal selections for the quantizers are plotted in Figures \ref{F:rb1_full} and \ref{F:rb3_full} respectively. 
\begin{figure}
    \centering
    \includegraphics[draft=false, width=0.7\textwidth]{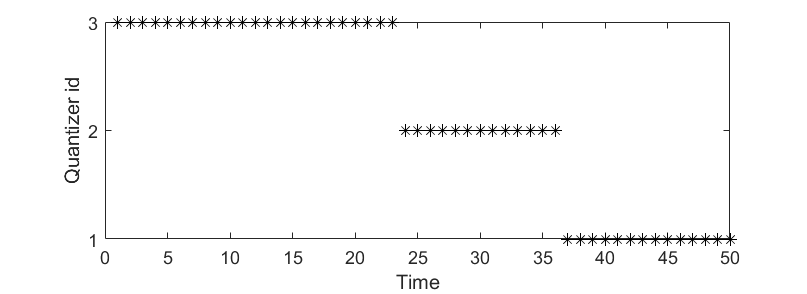}
    \caption{Optimal selection of the quantizers under perfect observation when $d_i=i$ for all $i$.}
    \label{F:rb1_full}
\end{figure}
\begin{figure}
    \centering
    \includegraphics[draft=false, width=0.7\textwidth]{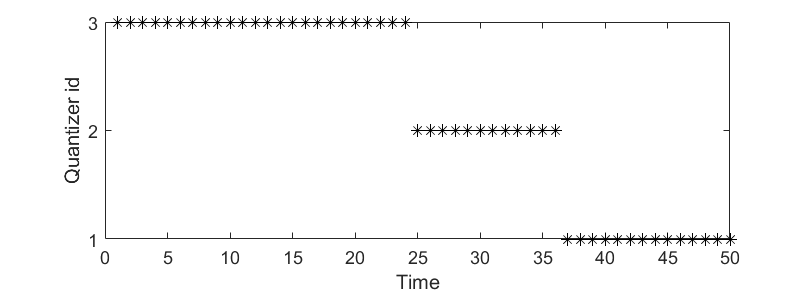}
    \caption{Optimal selection of the quantizers under perfect observation when $d_i=1$ for all $i$.}
    \label{F:rb3_full}
\end{figure}
In this perfect observation case, the system is not as keen in using the finest resolution quantization as it was for the noisy observation case; 
nonetheless, we still observe the dual effect of the quantization resolution and the associated delays in the optimal choice of the quantizers. 

\section{Conclusions} \label{S:conclusion}

In this work, we have considered a  quantization-based partially observed LQG problem with a   quantization cost. 
The problem is to choose an optimal quantizer among a set of available quantizers that minimizes the combined cost of quantization and control performance. 
The number of bits required to represent the quantized value increases as the quantization resolution gets better, and hence the delay transmitting the measurement also increases. We illustrate how the quality of quantization and delay together emerge in the cost function and we demonstrate their dual role in the optimal solution.

We have shown that the optimal controller exhibits a separation principle and it has a linear relationship with the estimate of the state. 
The optimal gains for the controller are found by solving the classical Riccati equation associated with the LQG problem. 
We have also shown that the optimal selection of the quantizers can be found by solving a linear program that can be solved offline independently.
 Furthermore, the special cases of full observation and constant delay are also discussed. 
 The possibility of the system to remain open-loop at time $t$ by not sending any quantized information,  is discussed as well in Remark \ref{R:openloop}.

The analysis of this paper relies on the idea of quantization of the innovation signal. As a future work it would be interesting to extend the similar idea beyond LQG systems.


\appendix

\section{Proof of Proposition \ref{Pr:gaussian}} \label{AP:1}

Let us consider a state-process $\xn_t$ and an observation-process $\yn_t$ as follows 
\begin{subequations} \label{E:newProcess}
\begin{align}
\xn_t&=X_t-\sum_{k=0}^{t-1}A^{t-1-k}BU_k-A^t\mu_0,\\
\yn_t&=C\xn_t+V_t.
\end{align}
\end{subequations} 
It follows that
\begin{subequations} \label{E:newProcess1}
\begin{align}
\xn_{t+1}&=A\xn_t+W_t,\\
\yn_t&=C\xn_t+V_t,\\
\xn_0&=X_0-\mu_0=W_{-1}\sim \N(0,\Sigma_x).
\end{align}
\end{subequations}
Here $\xn_t$ is the process associated with $X_t$, which is independent of the control strategy. Using this definition of $\xn_t$ and $\yn_t$, we have $X_t=\xn_t+\varphi(t,\U_{t-1})$ and $Y_t=\yn_t+C\varphi(t,\U_{t-1})$ where $\varphi(t,\U_{t-1})=\sum_{k=0}^{t-1}A^{t-1-k}BU_k+A^t\mu_0$.
Therefore, the information sets $(\Y_{t-1},\U_{t-1})$ and $(\yn_0,\ldots,\yn_{t-1},\U_{t-1})$ are equivalent, i.e., one can be constructed from the other.

The innovation process associated with system \eqref{E:newProcess1} is given by 
\begin{align*}
\xi^{\rm{new}}_t=\yn_t-\E[\yn_t|\yn_0,\ldots,\yn_{t-1}].
\end{align*}
Since $\xi_t$ is the innovation process associated with the system \eqref{E:dyn}, it can be shown that $\xi^{\rm{new}}_t=\xi_t$ for all $t$. In order to prove this statement, notice that
\begin{align*}
\xi_t=&Y_t-\E[Y_t|\Y_{t-1},\U_{t-1}]\\
=&\yn_t+C\varphi(t,\U_{t-1})-\E[\yn_t|\Y_{t-1},\U_{t-1}]-\E[C\varphi(t,\U_{t-1})|\Y_{t-1},\U_{t-1}]\\
=&\yn_t-\E[\yn_t|\yn_0,\ldots,\yn_{t-1},\U_{t-1}]\\=&\yn_t-\E[\yn_t|\yn_0,\ldots,\yn_{t-1}]=\xi^{\rm{new}}_t.
\end{align*}
Thus, $\xi_t$ does not depend on the control history $\U_{t-1}$.

The standard results of Kalman filtering hold for the process $\xn_t$ with observation $\yn_t$.
It follows that $\{\xi^{\rm{new}}_t\}_{t\in \bbN_0}$ is a sequence of uncorrelated Gaussian noises. 
Thus, using
standard Kalman filtering theory, we define 
\begin{subequations}\label{E:equations}
\begin{align}
\en_t=&\xn_t-\E[\xn_t|\yn_0,\ldots,\yn_{t-1}],\\
\Delta^{\rm{new}}_t=&\xn_t-\E[\xn_t|\yn_0,\ldots,\yn_{t}],\\
\Sigma_{t|t-1}=&\E[\en_t{\en_t}\T],\\
\Sigma_t=&\E[\Delta^{\rm{new}}_t{\Delta^{\rm{new}}_t}\T].
\end{align}
\end{subequations}
 
Moreover, 
\begin{align*}
\E[\xn_t|\yn_0,\ldots,\yn_{t}]=\E[\xn_t|\yn_0,\ldots,\yn_{t-1}]+K_t\xi^{\rm{new}}_t,
\end{align*}
where $K_t$ is the Kalman gain at time $t$.
 Thus, $\Delta^{\rm{new}}_t=\en_t-K_t\xi^{\rm{new}}_t=(I-K_tC)\en_t-K_tV_t$. The initial conditions are $\en_0=\xn_0\sim \N(0,\Sigma_x)$ and $\Sigma_{0|-1}=\Sigma_x.$
Therefore, $\E[\xi^{\rm{new}}_t]=0$ and $M_{t},\Sigma_{t|t-1}$ and $\Sigma_t$ satisfy 
\begin{align*}
M_t=&\E[(C\en_t+V_t)(C\en_t+V_t)\T]\\
=&C\Sigma_{t|t-1}C\T+\V,\\
\Sigma_{t|t-1}=&\E[\en_t{\en_t}\T]\\
=&\E[(A\Delta^{\rm{new}}_{t-1}+W_{t-1})(A\Delta^{\rm{new}}_{t-1}+W_{t-1})\T]\\
=&A\Sigma_{t-1}A\T+\W,\\
\Sigma_t=&\E[(I-K_tC)\en_t{\en_t}\T(I-K_tC)\T]+K_tC\V C\T K_t\T\\
=&(I-K_tC)\Sigma_{t|t-1}(I-K_tC)\T+K_t\V K_t\T\\
=&\Sigma_{t|t-1}-\Sigma_{t|t-1}C\T M_t^{-1} C\Sigma_{t|t-1},
\end{align*}
where $K_t=\Sigma_{t|t-1}C\T M_t^{-1}$ is the Kalman gain. This concludes the proof.\hfill $\blacksquare$

\section{Proof of Lemma \ref{L:inno}} \label{AP:inno}
Note that the information contained in $(\Y_t,\U_{t-1})$ is the same as the information contained in $(\Xi_t,\U_{t-1})$, where $\Xi_t=\{\xi_0,\ldots,\xi_t\}$. Therefore,
\begin{align*}
\E[X_t|\Y_t,\U_{t-1}]=&\E[X_t|\Xi_t,\U_{t-1}]\\
=&\E[\xn_t|\Xi_t,\U_{t-1}]+\sum_{k=0}^{t-1}A^{t-1-k}BU_k+A^t\mu_0\\
=&\E[\xn_t|\Xi^{\rm{new}}_t]+\sum_{k=0}^{t-1}A^{t-1-k}BU_k+A^t\mu_0,
\end{align*}
where $\Xi^{\rm{new}}_t=\{\xi^{\rm new}_t\}_{t \in \bbN_0}=\{\xi_t\}_{t \in \bbN_0}=\Xi_t$. From the theory of Kalman filtering, it follows that
\begin{align*}
\E[\xn_t|\Xi^{\rm{new}}_t]=&\E[\xn_t|\Xi^{\rm{new}}_{t-1}]+K_t\xi^{\rm new}_t\\
=&A\E[\xn_{t-1}|\Xi^{\rm{new}}_{t-1}]+K_t\xi^{\rm new}_t,
\end{align*}
since $W_{t-1}$ is independent of $\Xi^{\rm{new}}_{t-1}$. We need to show that 
\begin{align} \label{E:innohypo}
\E[\xn_t|\Xi^{\rm{new}}_t]=\sum_{k=0}^t\Psi(t,k)\xi_k^{\rm new},
\end{align}
for some $\Psi(t,k)$. We show this by induction. To this end,
notice that \eqref{E:innohypo} is true for $t=0$ with $\Psi(0,0)= \Sigma_x C\T (C\Sigma_x C\T+\V)^{-1}$, where $\Sigma_x$ is the covariance of the initial state $X_0$.
Next, if \eqref{E:innohypo} is true for $t=\tau$, then we have that, for $t=\tau+1$,
\begin{align*}
\E[\xn_{\tau+1}|\Xi^{\rm{new}}_{\tau+1}]=&A\E[\xn_{\tau}|\Xi^{\rm{new}}_{\tau}]+K_{\tau+1}\xi^{\rm new}_{\tau+1}\\
=&A\sum_{k=0}^{\tau}\Psi(\tau,k)\xi_k^{\rm new}+K_{\tau+1}\xi^{\rm new}_{\tau+1}\\
=&\sum_{k=0}^{\tau+1}\Psi(\tau+1,k)\xi_k^{\rm new},
\end{align*}
where $K_{\tau+1}$ is the Kalman gain at time $\tau+1$, $\Psi(\tau+1,k)=A\Psi(\tau,k)$ for all $k=0,\ldots,\tau$, and $\Psi(\tau+1,\tau+1)=K_{\tau+1}$. Therefore, for all $t\ge k$, $\Psi(t,k)=A^{t-k}K_k=A^{t-k}\Sigma_{k|k-1}C\T M_k^{-1}$, and \vspace{-5pt}
\begin{align} \label{E:X_given_Y_U}
\begin{split}
\E[X_t|\Y_t,\U_{t-1}]=&\E[\xn_t|\Xi^{\rm{new}}_t]+\sum_{k=0}^{t-1}A^{t-1-k}BU_k+A^t\mu_0\\
=&\sum_{k=0}^t\Psi(t,k)\xi_k^{\rm new}+\sum_{k=0}^{t-1}A^{t-1-k}BU_k+A^t\mu_0\\
=&\sum_{k=0}^t\Psi(t,k)\xi_k+\sum_{k=0}^{t-1}A^{t-1-k}BU_k+A^t\mu_0.
\end{split}
\end{align}
 The set $\Z_t$ may not contain all the elements of $\Xi_t$ due to delays.
 In fact, for $k\le t$, we have that $\xi_k \in \Z_t$ if and only if $\vartheta_{k,t}=1$. Since $\xi_k$ and $\xi_t$ are independent for $t\ne k$, we have
\begin{align*}
\E[\xi_k|\Z_t]=\begin{cases}
\xi_k, ~~~~ &{\textrm{if   }} \xi_k \in \Z_t,\\
0,~~~&{\rm otherwise}.
\end{cases}
\end{align*}
 Therefore, we can write $\E[\xi_k|\Z_t]=\vartheta_{k,t}\xi_k$. Thus,
 \begin{align*}
 \E[X_t|\Z_t,\U_{t-1}]=&\E[\E[X_t|\Xi_t,\U_{t-1}]|\Z_t,\U_{t-1}]\\
 =&\E\left[\sum_{k=0}^t\Psi(t,k)\xi_k|\Z_t,\U_{t-1}\right]+\sum_{k=0}^{t-1}A^{t-1-k}BU_k+A^t\mu_0\\
 =&\sum_{k=0}^t\Psi(t,k)\vartheta_{k,t}\xi_k+\sum_{k=0}^{t-1}A^{t-1-k}BU_k+A^t\mu_0.
 \end{align*}
 This completes the proof.
\hfill $\blacksquare$
 
 \section{Proof of Lemma \ref{L:et_second_moment}} \label{A:delta}

 Let us define $\Delta_t = \E[e_t \mid \Y_t,\hat{\mO}_t, \U_{t-1}]$, and
 notice that,
 \begin{align*}
     \E[e_t e_t\T \mid \Y_t,\hat{\mO}_t, \U_{t-1}] = \E[(e_t - \Delta_t)(e_t -\Delta_t)\T \mid \Y_t,\hat{\mO}_t, \U_{t-1}] +\E[\Delta_t \Delta_t\T \mid \Y_t,\hat{\mO}_t, \U_{t-1}],
 \end{align*}
 since $\E[\Delta_t(e_t-\Delta)\T \mid \Y_t,\hat{\mO}_t, \U_{t-1}]= \Delta_t\E[(e_t-\Delta)\T \mid \Y_t,\hat{\mO}_t, \U_{t-1}] = 0$.
 Taking expectations on both sides of the last equation, we obtain
 \begin{align} \label{E:et_et_transpose}
     \E[e_te_t\T] =  \E[(e_t - \Delta_t)(e_t -\Delta_t)\T] +\E[\Delta_t \Delta_t\T].
 \end{align}

 Substituting the expression of $\bar X_t$ from \eqref{E:x_bar_new} in  $e_t =X_t - \bar X_t$, yields
 \begin{align*}
     e_t = X_t -\sum_{k=0}^{t-1}A^{t-1-k}BU_k - A^t\mu_0 - \sum_{k=0}^t\Psi(t,k)\vartheta_{k,t}\bar\xi_k.
 \end{align*}
 Therefore,
 \begin{align} \label{E:delta}
 \begin{split}
     \Delta_t =& \E[e_t\mid \Y_t,\hat{\mO}_t, \U_{t-1}] \\
     =&\E[X_t\mid \Y_t,\hat{\mO}_t, \U_{t-1}] - \sum_{k=0}^{t-1}A^{t-1-k}BU_k - A^t\mu_0 - \sum_{k=0}^t\Psi(t,k)\vartheta_{k,t}\bar\xi_k \\
     =& \sum_{k=0}^t\Psi(t,k) \xi_k - \sum_{k=0}^t\Psi(t,k)\vartheta_{k,t}\bar\xi_k \\
     =& \sum_{k=0}^t\Psi(t,k) ( \xi_k - \E[\xi_k \mid \hat \mO_t]) 
     \end{split}
 \end{align}
 where we have used $\E[X_t\mid \Y_t,\hat{\mO}_t,\U_{t-1}] = \E[X_t\mid \Y_t, \U_{t-1}]$ since $\hat{\mO}_t$ is a $\Y_t$-measurable function and we have also used \eqref{E:xi_expected} to write $\vartheta_{k,t}\bar\xi_k$ as $\E[\xi_k \mid \hat \mO_t]$.
  Using the expression of $\Delta_t$ from \eqref{E:delta}, we obtain
 \begin{align*}
     e_t -\Delta_t =&   X_t -\sum_{k=0}^{t-1}A^{t-1-k}BU_k - A^t\mu_0 -\sum_{k=0}^t\Psi(t,k) \xi_k \\
     =& \xn_t - \E[\xn_t\mid \Xi^{\rm{new}}_t] = \Delta^{\rm{new}}_t,
 \end{align*}
 where $\xn_t$, $ \E[\xn_t\mid \Xi^{\rm{new}}_t]$, and $\Delta^{\rm{new}}_t$ are defined in equations \eqref{E:newProcess} and \eqref{E:innohypo}, and \eqref{E:equations}, respectively.
  Thus, we may rewrite \eqref{E:et_et_transpose} as follows
 \begin{align}
 \begin{split} \label{eq:et_et_transpose_2}
     \E[&e_te_t\T] =  \E[(e_t - \Delta_t)(e_t -\Delta_t)\T] +\E[\Delta_t \Delta_t\T] \\
     &= \E[\Delta^{\rm{new}}_t\!{\Delta^{\rm{new}}_t}\T]\! +\!\! \sum_{k=0}^t\!\sum_{\ell =0 }^t \!\Psi(t,k) \E[( \xi_k - \E[\xi_k \mid \hat \mO_t])( \xi_\ell - \E[\xi_\ell \mid \hat \mO_t])\T]\Psi(t,\ell)\T\\
     &= \Sigma_t\! +\!\! \sum_{k=0}^t\!\sum_{\ell =0 }^t \!\Psi(t,k) \E[( \xi_k - \E[\xi_k \mid \hat \mO_t])( \xi_\ell - \E[\xi_\ell \mid \hat \mO_t])\T]\Psi(t,\ell)\T
     \end{split}
 \end{align}
where we used the definition $\Sigma_t=\E[\Delta^{\rm{new}}_t{\Delta^{\rm{new}}_t}\T]$ from \eqref{E:equations}.

To further simplify \eqref{eq:et_et_transpose_2}, we recall that $\xi_k$ and $\xi_\ell$ are independent random variables when $k\ne \ell$ and therefore, we obtain
\begin{align*}
    \E[( \xi_k - \E[\xi_k \mid \hat \mO_t])( \xi_\ell - \E[\xi_\ell \mid \hat \mO_t])\T] =& \E \big[\E[( \xi_k - \E[\xi_k \mid \hat \mO_t])( \xi_\ell - \E[\xi_\ell \mid \hat \mO_t])\T\mid \xi_k, \hat \mO_t]\big] \\
    =&\E \big[( \xi_k - \E[\xi_k \mid \hat \mO_t]) \E[( \xi_\ell - \E[\xi_\ell \mid \hat \mO_t])\T\mid \xi_k, \hat \mO_t] \big] \\
    =& \E [( \xi_k - \E[\xi_k \mid \hat \mO_t]) ( \E[\xi_\ell\mid \xi_k, \hat \mO_t] - \E[\xi_\ell \mid \hat \mO_t])\T] \\
    =& \E [( \xi_k - \E[\xi_k \mid \hat \mO_t]) ( \E[\xi_\ell \mid \hat \mO_t] - \E[\xi_\ell \mid \hat \mO_t])\T]  = 0
\end{align*}
for all $k \ne \ell$. On the other hand, for $k=\ell$, we obtain
\begin{align*}
     \E[( \xi_k - \E[\xi_k \mid \hat \mO_t])( \xi_k - \E[\xi_k \mid \hat \mO_t])\T] =& \E[\xi_k\xi_k\T] - \E[ \E[\xi_k \mid \hat \mO_t] \E[\xi_k \mid \hat \mO_t]\T] \\
     \overset{(a)}{=}& M_k - \E[\vartheta_{k,t}\bar \xi_k\bar \xi_k\T] \\
     =& M_k - \E[\E[\vartheta_{k,t}\bar \xi_k\bar \xi_k\T\mid \theta_k]] \\
     \overset{(b)}{=}& M_k -\E[\vartheta_{k,t}\E[\bar \xi_k\bar \xi_k\T\mid \theta_k]]\\
     \overset{(c)}{=}& M_k -\E[\vartheta_{k,t}F_k(\theta_k)]
\end{align*}
where (a) follows from \eqref{E:xi_expected} and the fact that $\vartheta_{k,t}^2 = \vartheta_{k,t}$ since $\vartheta_{k,t} \in \{0,1\}$, and (b) follows from the fact that $\vartheta_{k,t}$ is a deterministic function of $\theta_k$ due to \eqref{E:vartheta}, and finally, (c) follows from  \eqref{E:F}.
Consequently, \eqref{eq:et_et_transpose_2} reduces to the following equation
\begin{align*}
    \begin{split} \label{eq:et_et_transpose_3}
     \E[&e_te_t\T] =  \Sigma_t\! +\!\! \sum_{k=0}^t \!\Psi(t,k) (M_k -\E[\vartheta_{k,t}F_k(\theta_k)])\Psi(t,k)\T.
     \end{split} ~~~~~~~~~~~\hfill \blacksquare
\end{align*}

\bibliographystyle{siamplain}
\bibliography{arXiv}

\end{document}